\documentclass[12pt]{article}
\usepackage[utf8]{inputenc}
\usepackage{amsthm,amsmath,amssymb,bbm,hyperref,graphicx,float,tikz}
\usepackage{xcolor}
\usepackage[normalem]{ulem}
\allowdisplaybreaks

\newtheorem{thm}{Theorem}
\newtheorem{prop}{Proposition}
\newtheorem{lemma}{Lemma}
\newtheorem{cor}{Corollary}
\theoremstyle{definition}
\newtheorem{defn}{Definition}
\theoremstyle{remark}
\newtheorem{rmk}{Remark}
\newtheorem{asm}{Assumption}

\DeclareMathOperator{\tr}{tr}

\newcommand{\RE}{\textup{Re\,}}
\newcommand{\IM}{\textup{Im\,}}
\newcommand{\Hilbert}{\mathcal{H}}
\newcommand{\CCC}{\mathbb{C}}
\newcommand{\EEE}{\mathbb{E}}
\newcommand{\PPP}{\mathbb{P}}
\newcommand{\RRR}{\mathbb{R}}
\newcommand{\SSS}{\mathbb{S}}
\newcommand{\be}{\begin{equation}}
\newcommand{\ee}{\end{equation}}

\newcommand{\scp}[2]{\langle #1|#2 \rangle}

\title{Typical Macroscopic Long-Time Behavior for Random Hamiltonians}
\author{
Stefan Teufel\footnote{Mathematics Institute, Eberhard Karls University T\"ubingen, 
	Auf der Morgenstelle 10, 72076 T\"ubingen, Germany. ORCID: 0000-0003-3296-4261,
	E-mail: stefan.teufel@uni-tuebingen.de},~
Roderich Tumulka\footnote{Mathematics Institute, Eberhard Karls University T\"ubingen, 
	Auf der Morgenstelle 10, 72076 T\"ubingen, Germany. 
	ORCID: 0000-0001-5075-9929
	E-mail: roderich.tumulka@uni-tuebingen.de},~ 
Cornelia Vogel\footnote{Mathematics Institute, Eberhard Karls University T\"ubingen, 
	Auf der Morgenstelle 10, 72076 T\"ubingen, Germany.
	ORCID: 0000-0002-3905-4730,
	E-mail: cornelia.vogel@uni-tuebingen.de}
}
\date{November 8, 2024}

\begin{document}

\maketitle

\begin{abstract}
We consider a closed macroscopic quantum system in a pure state $\psi_t$ evolving unitarily and take for granted that different macro states correspond to mutually orthogonal subspaces $\mathcal{H}_\nu$ (macro spaces) of Hilbert space, each of which has large dimension. We extend previous work on the question what the evolution of $\psi_t$ looks like macroscopically, specifically on how much of $\psi_t$ lies in each $\mathcal{H}_\nu$. Previous bounds concerned the \emph{absolute} error for typical $\psi_0$ and/or $t$ and are valid for arbitrary Hamiltonians $H$; now, we provide bounds on the \emph{relative} error, which means much tighter bounds, with probability close to 1 by modeling $H$ as a random matrix, more precisely as a random band matrix (i.e., where only entries near the main diagonal are significantly nonzero) in a basis aligned with the macro spaces. We exploit particularly that the eigenvectors of $H$ are delocalized in this basis. Our main mathematical results confirm the two phenomena of generalized normal typicality (a type of long-time behavior) and dynamical typicality (a type of similarity within the ensemble of $\psi_0$ from an initial macro space). They are based on an extension we prove of a no-gaps delocalization result for random matrices by Rudelson and Vershynin \cite{RV16}. 

\medskip

\noindent Key words: normal typicality; dynamical typicality; eigenstate thermalization hypothesis; delocalized eigenvector; macroscopic quantum system.
\end{abstract}

\section{Introduction}

We are concerned with an isolated quantum system in a pure state $\psi$ comprising a macroscopically large number of particles. Studying such systems in order to study macroscopic behavior is an approach that is particularly used in connection with the \textit{eigenstate thermalization hypothesis} (ETH) which has been investigated by physicists as well as mathematicians \cite{Deutsch91,Srednicki94,Erdoes21,CEH23}. In particular, the phenomena of equilibration and thermalization in closed quantum systems have attracted increasing attention in recent years \cite{BRGSR18,BG09,GHT13,GHT14,GHT15,GLMTZ09,GLTZ10,Reimann08,Reimann2015,Reimann2018b,Reimann2018a,RG20,Short11,SF12,SWGW22,TTV22-physik}.

We assume that the Hilbert space $\mathcal{H}$ of the system has very large but finite dimension.\footnote{The physical background is that $\Hilbert$ is really the subspace of the full Hilbert space corresponding to a micro-canonical energy interval $[E-\Delta E, E]$ with $\Delta E$ the resolution of macroscopic energy measurements, and this interval contains a very large but finite number of eigenvalues of the Hamiltonian, realistically of the order $10^{10^{10}}$. However, in this paper we will not assume that all eigenvalues of $H$ on $\Hilbert$ lie between $E-\Delta E$ and $E$.}
Following Wigner \cite{Wigner55}, we model the Hamiltonian as a Hermitian random matrix~$H$. Following von Neumann \cite{vonNeumann29}, we regard as given an orthogonal decomposition
\begin{align}
    \Hilbert = \bigoplus_{\nu}\Hilbert_\nu \label{eq: decomp}
\end{align}
of the system's Hilbert space $\Hilbert$ into subspaces $\Hilbert_\nu$ (``macro spaces'') associated with different macro states $\nu$ \cite{Leb93,GLTZ10,SWGW22,TTV22-physik}. The macroscopic behavior or macroscopic appearance of a vector $\psi\in\Hilbert$ is then represented by the weights given by $\psi$ to different macro states $\nu$, i.e., by the sizes of the contributions to $\psi$ from the various $\Hilbert_\nu$, $\|P_\nu\psi\|^2$ with $P_\nu$ the projection to $\Hilbert_\nu$. Here we ask how, under the unitary time evolution $\psi_t = \exp(-iHt)\psi_0$, the macroscopic appearance of $\psi_t$ typically evolves and equilibrates in the long run, more precisely, what $\|P_\nu\psi_t\|^2$ are for typical $\psi_0$ from a macro space and/or typical large $t$ (where ``typical $\psi_0$'' refers to the uniform distribution over the unit sphere; more precise formulations later). Thereby, we extend and refine previous work \cite{TTV22-physik} on the same physical question by means of a deeper mathematical analysis of the behavior of $\|P_\nu\psi_t\|^2$ for typical $H$ (on top of typical $\psi_0$ and typical $t$) from suitable random matrix ensembles.

\subsection{Normal Typicality}

The macroscopic appearance of a system with $\psi\in\Hilbert_\nu$ is summarized by the macro state~$\nu$; of course, a general state is a superposition of contributions from different macro spaces, which makes the superposition weights $\|P_\nu \psi\|^2$ relevant. It can be useful to compare a given state $\psi$ to a \emph{purely random} vector~$\phi$, i.e., one that is uniformly distributed over the unit sphere $\mathbb{S}(\Hilbert) = \{\psi\in\Hilbert: \|\psi\|=1\}$; $\phi$ has, with probability near 1, superposition weights approximately proportional to the dimension of $\Hilbert_\nu$,
\begin{align}
    \|P_\nu\phi\|^2 \approx \frac{d_\nu}{D},
\end{align}
provided that $d_\nu:= \dim\Hilbert_\nu$ and $D:=\dim\Hilbert$ are large; see, e.g., Lemma 1 in \cite{GLMTZ09} and the references therein. Such a behavior is sometimes called ``normal,'' in analogy to the concept of a normal number \cite{Normalnumber}.

In the cases that von Neumann considered in \cite{vonNeumann29}, it turned out that every $\psi_0\in\mathbb{S}(\Hilbert)$ evolves so that for most $t\in[0,\infty)$,
\begin{align}
    \|P_\nu\psi_t\|^2 \approx \frac{d_\nu}{D} \quad \forall \nu \label{eq: NT}
\end{align}
if $d_\nu$ and $D$ are sufficiently large. This behavior is called ``normal typicality'' and more elaborated in \cite{GLMTZ09, GLTZ10}. 

However, the cases von Neumann considered are not very realistic: His conditions on~$H$ are true with probability near 1 if the eigenbasis of $H$ is chosen purely randomly (i.e., uniformly) among all orthonormal bases (and some further technical conditions that are not very restrictive). This can be regarded as expressing that the energy eigenbasis is \textit{unrelated} to the orthogonal decomposition \eqref{eq: decomp}. In this case, the system would very rapidly go from any initial macro space $\Hilbert_\mu$ to the thermal equilibrium macro space $\Hilbert_\textup{eq}$ \cite{GHT13,GHT14,GHT15}, which is unphysical.\footnote{Thermal equilibrium requires that energy (and other quantities) is rather evenly distributed over all degrees of freedom, and for getting an even distribution, it needs to get transported through space, which requires time and passage through other macro states.} That is why we are interested in generalizations of normal typicality that apply also to Hamiltonians whose eigenbasis is not unrelated to the decomposition \eqref{eq: decomp}.
 
In such more general scenarios, we showed in \cite{TTV22-physik} that the following \emph{generalized normal typicality} still holds: for most $\psi_0$ from the unit sphere $\mathbb{S}(\Hilbert_\mu)$ in the macro space associated with a (possibly non-equilibrium) macro state $\mu$ and most $t\in[0,\infty)$,
\begin{align}
    \|P_\nu\psi_t\|^2 \approx M_{\mu\nu} \quad \forall\nu \label{eq: GNT}
\end{align}
for suitable values $M_{\mu\nu}$, in fact (for non-degenerate $H$)
\be\label{Mmunuexpression}
M_{\mu\nu}= \frac{1}{d_\mu} \sum_n \langle\phi_n|P_\mu|\phi_n\rangle \langle\phi_n|P_\nu|\phi_n\rangle 
\ee
with $\phi_1,\ldots,\phi_D$ an orthonormal eigenbasis of $H$. 
Put another way, generalized normal typicality means that the superposition weights $\|P_\nu\psi_t\|^2$ approach certain stable values and stay close to them for most times; we also say that the system \textit{equilibrates}. See Figure~\ref{fig:numerical example 1} and Figure~\ref{fig:numerical example 2} for a numerical example.

\begin{figure}[h]
    \centering
    \includegraphics[width = 0.7\textwidth]{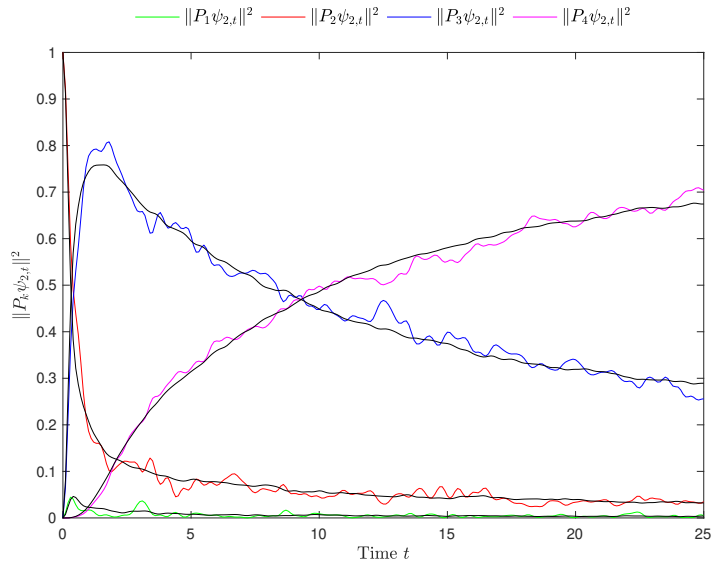}
    \caption{Numerical simulation of the functions  $t\mapsto\|P_\nu\psi_t\|^2$ as detailed in Appendix~\ref{sec: numerics}. A Hilbert space of dimension $D=2222$ is decomposed into 4 macro spaces of dimensions $d_1=2$ (green curve), $d_2=20$ (red curve), $d_3=200$ (blue curve), and $d_4=2000$ (purple curve). At large times, see also Figure~\ref{fig:numerical example 2}, the equilibrium subspace $\Hilbert_4$ has the biggest contribution to $\psi_t$. The initial wave function $\psi_0$ was chosen purely randomly from the unit sphere $\mathbb{S}(\Hilbert_2)$, and the Hamiltonian is a random band matrix in a basis aligned with the macro spaces with a band that is wide enough such that the eigenfunctions are still delocalized. Then parts of $\psi_t$ reach $\Hilbert_4$ only after passing through the macro space $\Hilbert_3$; in this example the blue curve increases first before it decreases and the purple curve increases. The black curves are the deterministic approximations $w_{2\nu}(t)$, see \eqref{eq: wmuB} according to dynamical typicality. (This figure already appeared in \cite{TTV22-physik}.)}
    \label{fig:numerical example 1}
\end{figure}

\begin{figure}[h]
    \centering
    \includegraphics[width = 0.7\textwidth]{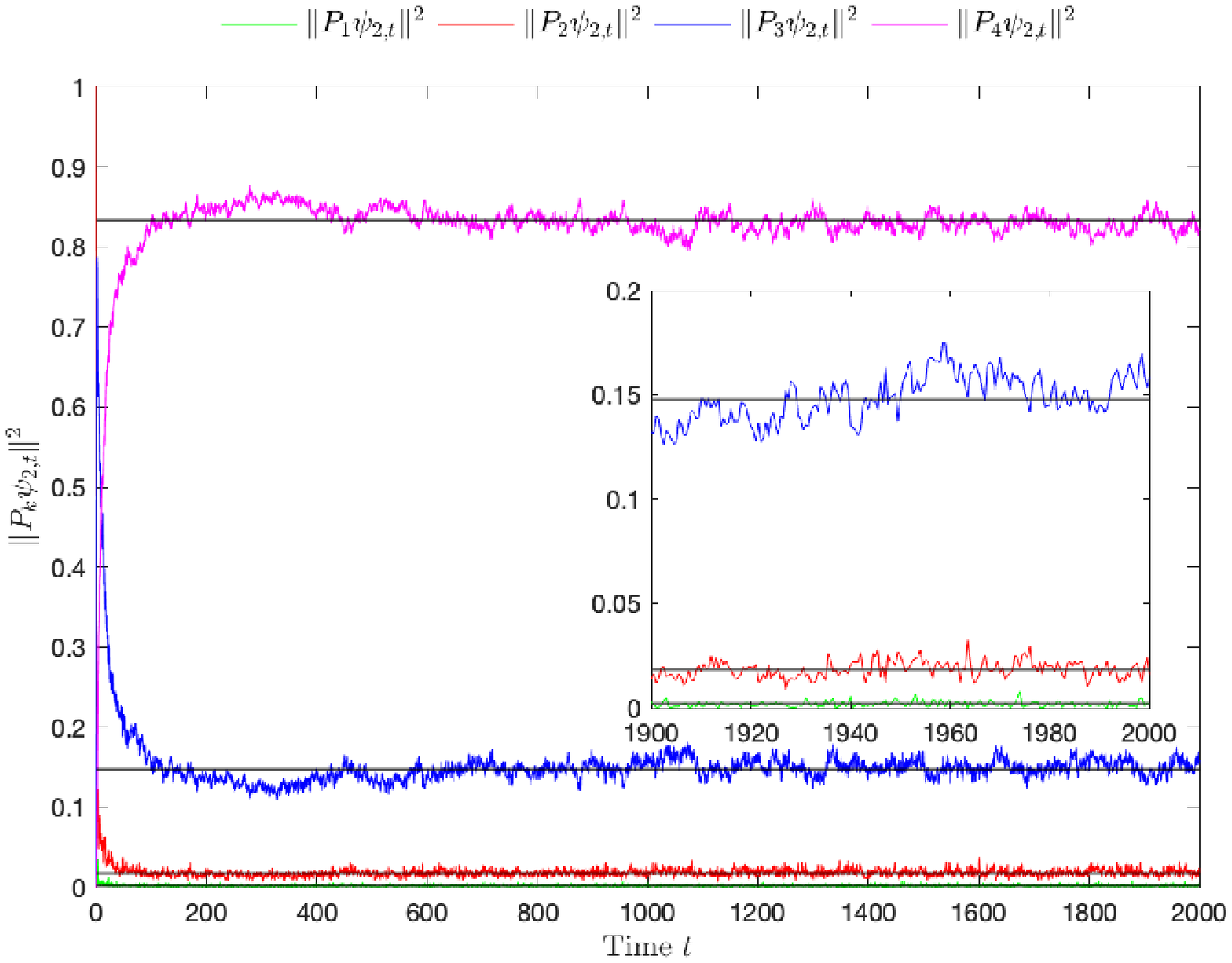}
    \caption{The same simulation as in Figure~\ref{fig:numerical example 1}, only for larger times. The inset shows a part of the figure in magnification. In the long run, the curves $t\mapsto \|P_\nu\psi_t\|^2$ reach 
    the values $M_{\mu\nu}$ (indicated by the black horizontal lines) and stay close to them, up to either small or rare fluctuations. (This figure also already appeared in \cite{TTV22-physik}.)}
    \label{fig:numerical example 2}
\end{figure}

One of our goals in this paper is to strengthen the results for suitable types of random Hamiltonians. We are particularly interested in  random matrices with band structure (see Figure~\ref{fig:bandmatrix}) in a basis that diagonalizes the projections onto the macro spaces. This band structure makes it less likely that a state in a small macro space (far away from equilibrium) goes directly into the thermal equilibrium macro space without passing through some other macro spaces in between.

\begin{figure}[h]
\begin{center}
\begin{tikzpicture}
\filldraw[color=gray] (6,0)--(6,0.3)--(0.3,6)--(0,6)--(0,5.7)--(5.7,0)--cycle;
\draw (0,0)--(6,0)--(6,6)--(0,6)--cycle;
\draw (0,6)--(6,0);
\draw (2,6)--(2,0);
\draw (6,4)--(0,4);
\draw (1,6)--(1,0);
\draw (6,5)--(0,5);
\draw (0.5,6)--(0.5,0);
\draw (6,5.5)--(0,5.5);
\end{tikzpicture}
\end{center}
\caption{A band matrix has significantly nonzero entries only in the grey region near the main diagonal. The first macro space is spanned by the first $d_1$ basis vectors, the second macro space by the next $d_2$ basis vectors, etc.; in the figure, a line is drawn after the first $d_1$ columns and the first $d_1$ rows, etc.}
\label{fig:bandmatrix}
\end{figure}
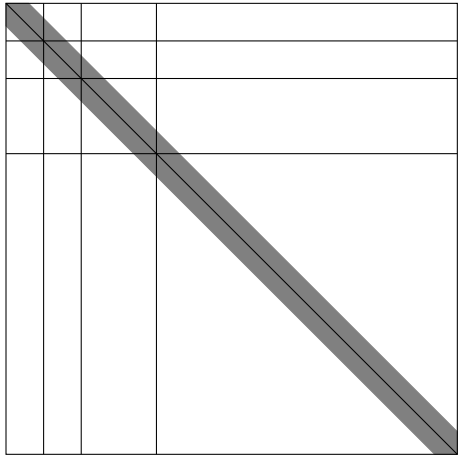

We showed in \cite{TTV22-physik} that for general Hamiltonians under suitable assumptions the \textit{absolute} errors 
\be
\bigl| \|P_\nu \psi_t\|^2 -M_{\mu\nu} \bigr|
\ee
are small (see Theorems \ref{thm: abs errors A}, \ref{thm: GNT} in Section~\ref{sec: prelim}). However, if 
$d_\nu \ll D$ then
we expect that $\|P_\nu\psi_t\|^2\ll 1$ and $M_{\mu\nu}\ll 1$. Then the absolute error 
would be small indeed, but this would not imply that $\|P_\nu \psi_t\|^2/M_{\mu\nu} \approx 1$. Therefore we want to show in the following that for suitable random Hamiltonians the \textit{relative} errors 
\be
\frac{\bigl| \|P_\nu \psi_t\|^2 -M_{\mu\nu}\bigr|}{M_{\mu\nu}}
\ee
are small as well. This is the first goal of this paper: to show for certain tractable models and under suitable conditions on the $d_\nu$ that \eqref{eq: GNT} holds in the sense that even the \emph{relative} error is small. And this goal is indeed highly relevant, since we expect that $d_\nu\ll D$ for \emph{all} non-equilibrium macrospaces $\Hilbert_\nu$.

Rigorous statements are formulated in Section \ref{sec: main results}: In Theorem \ref{thm: main new} we provide a lower bound on $M_{\mu\nu}$ under the assumption that $H=H_0+V$ where $H_0$ is any deterministic Hermitian matrix and $V$ is a Gaussian random matrix with entries that are independent up to Hermitian symmetry and have variances bounded away from 0 (so also far from the main diagonal, entries of $H$ cannot be \emph{too} small). From this lower bound and the upper bound on the absolute error provided in \cite{TTV22-physik} (and recapitulated in Section \ref{sec: prelim}), we obtain an upper bound on the relative error in Corollary \ref{cor: real dim rel}, valid with high probability (i.e., for most $H$) for most $\psi_0\in\SSS(\Hilbert_\mu)$ and most $t\in[0,\infty)$. That is, the upshot is that for typical $H$ from the ensemble considered, $\|P_\nu\psi_t\|^2$ is nearly independent of $\psi_0$ and in the long run nearly independent of $t$---it equilibrates.

This equilibration is related to the question of the increase of the quantum Boltzmann entropy observable
\begin{align}
    S_B := \sum_\nu S(\nu) P_\nu
\end{align}
with entropy values
\be
S(\nu):=k_B\log(d_\nu)\,,
\ee
where $k_B$ is the Boltzmann constant. Here is how: Physically, in every energy shell corresponding to a small energy interval $[E-\Delta E,E]$, there usually is one macro space $\Hilbert_\nu$, the one corresponding to thermal equilibrium, that has the overwhelming majority of dimensions, $\frac{d_\nu}{D}\approx 1$ \cite{GLMTZ10}; we will denote it by $\Hilbert_{\textup{eq}}$. Now if normal typicality holds as in \eqref{eq: NT}, or already if the $M_{\mu\nu}$ in \eqref{eq: GNT} tend to be bigger for $\nu$ with bigger dimensions $d_\nu$, then an initial state $\psi_0\in\mathbb{S}(\Hilbert_\mu)$ from a non-equilibrium (i.e., comparatively small) macro space $\Hilbert_\mu$ will evolve so that at most times $t$, the biggest contribution to $\psi_t$ lies in $\Hilbert_{\textup{eq}}$, and if $\|P_{\textup{eq}}\psi_t\|^2 \approx 1$, this implies further that $\|S_B\psi_t\| \approx S(\textup{eq})\gg S(\mu)=\|S_B\psi_0\|$, that is, the state has evolved to one with higher quantum Boltzmann entropy.

In order to obtain lower bounds for the $M_{\mu\nu}$ (and therefore upper bounds for the relative errors), we need lower bounds on $\|P_\nu\phi_n\|$, where $(\phi_n)$ is again an orthonormal eigenbasis of $H$. 

We take as fixed an orthonormal basis $(|j\rangle)_{j=1}^D$ of $\Hilbert$ that diagonalizes the $P_\nu$ (i.e., is such that each $|j\rangle$ lies in some $\Hilbert_\nu$). When we talk of $H$ as a matrix, we refer to this basis. In terms of this basis, 
\be \label{Pnuphinj}
\|P_\nu \phi_n\|^2 = \sum_{j\in I_\nu} |\langle \phi_n|j\rangle|^2 
\ee 
with $I_\nu$ the appropriate subset of $[D]:= \{1,\ldots,D\}$. 
In order to obtain a lower bound on this quantity we ask whether $\phi_n$ (expressed as an element of $\CCC^D$ relative to the basis $(|j\rangle)_j$) is localized or delocalized.

Our first main result, Theorem~\ref{thm: main new}, is based on results by Rudelson and Vershynin \cite{RV16} about the \textit{delocalization of eigenvectors} of a random matrix. While the delocalization of eigenvectors of different types of random matrices has been studied intensively in the literature in the recent years, e.g., \cite{ESY09,ESY09b,EK11,EK11b,EYY12,RV15,YYY21,AEK17,EKS19,BYY20,CEHK23}, most results were concerned with delocalization in the sup-norm, meaning that 
\be
\|\phi_n\|_\infty = \sup_{j=1}^D |\langle \phi_n|j\rangle|
\ee
cannot be too large, so many entries of $\phi_n$ must be non-negligible. However, that would still allow a negligible $\|P_\nu \phi_n\|$ for some $\nu$. Thus, for a lower bound on \eqref{Pnuphinj}, we need that only few entries are negligible. For this kind of reasons, results about the sup-norm only yield upper bounds for the $M_{\mu\nu}$ and, in very special cases, lower bounds for some of the $M_{\mu\nu}$; for example in Section \ref{sec: Improved LB} we show that with a sup-norm delocalization result from Ajanki, Erdős, and Krüger (2017) \cite{AEK17} we obtain useful lower bounds (Theorem~\ref{thm: lb AEK}) only if $\mu$ or $\nu$ is the thermal equilibrium macro state and $d_{\textup{eq}}$ is extremely large. 
The result from Rudelson and Vershynin (2016) \cite{RV16}, however, concerns another aspect of the delocalization of eigenvectors: it rules out gaps, i.e., no significant fraction of the coordinates 
of an eigenvector can carry only a negligible fraction of its mass. Their result gives lower bounds on \eqref{Pnuphinj} and enables us to prove nontrivial lower bounds for all $M_{\mu\nu}$. Unfortunately, the lower bounds for $M_{\mu\nu}$ obtained in this way seem to be far from optimal. More precisely, as long as one is in the delocalized regime, one expects that actually $M_{\mu\nu} \approx \frac{d_\nu}{D}$, as for normal typicality, even though the eigenvectors are not uniformly distributed over the sphere. (A matrix with uniformly chosen eigenvectors will be unlikely to have band structure, but a band matrix can very well have $M_{\mu\nu}\approx \frac{d_\nu}{D}$, as the simple example of the discrete Laplacian in 1d shows.) The bounds obtained by Rudelson and Vershynin were recently improved for matrices with independent entries \cite{LR20,LT20} but since we are interested in Hermitian random matrices (since they can serve as Hamiltonians), these improved results are not applicable in our situation.

Our theorems in this paper are proved for a general deterministic Hermitian matrix $B$ instead of a projection $P_\nu$, in parallel to previous works such as \cite{Short11,SF12,RG20}; in combination with the results from \cite{TTV22-physik} we can also obtain a finite-time result (that concerns most $t\in[0,T]$ instead of most $t\in [0,\infty)$).

Sharper estimates for all $M_{\mu\nu}$ can be obtained for certain ensembles of Hermitian random matrices for which stronger statements about the delocalization of eigenvectors are available. However, such statements are apparently not currently available for random band matrices, but only for matrices that have very substantially nonzero entries also far from the diagonal. Nevertheless, we report here also what would follow about $M_{\mu\nu}$ for those matrices, based on two results in the literature. First, in Section \ref{sec: Improved LB}, we use a delocalization result of Ajanki, Erdős, and Krüger \cite{AEK17} to obtain that $M_{\mu\nu}\approx d_\nu/D$ in case $\nu=\textup{eq}$ or $\mu=\textup{eq}$. Second, in Section \ref{sec: ETH}, we consider a version of the ETH formulated by Cipolloni, Erdős, and Henheik (2023) \cite{CEH23} and show that it implies that $M_{\mu\nu}\approx d_\nu/D$ for \emph{all} $\mu$ and $\nu$.

\subsection{Dynamical Typicality}

In \cite{TTV22-physik} we were not only concerned with generalized normal typicality but also with \textit{dynamical typicality}, which means that for any given $t$, $\|P_\nu\psi_t\|^2$ is nearly independent of $\psi_0\in\SSS(\Hilbert_\mu)$. In Figure \ref{fig:numerical example 1}, this phenomenon is visible in the proximity of the colored (exact) curves to the black ones. Here, we also provide an improved result about dynamical typicality for random $H$.

The name ``dynamical typicality'' was introduced by Bartsch and Gemmer \cite{BG09} for the following phenomenon: Given an observable $A$ and some value $a\in\mathbb{R}$, there is a function $a(t)$ such that for every $t\in\mathbb{R}$ and most initial wave functions $\psi_0\in\mathbb{S}(\Hilbert)$ with $\langle\psi_0|A|\psi_0\rangle \approx a$, it holds that $\langle\psi_t|A|\psi_t\rangle\approx a(t)$. A rigorous version of this fact was proved by Müller, Gross and Eisert \cite{MGE} who considered initial wave functions $\psi_0\in\mathbb{S}(\Hilbert)$ with $\langle\psi_0|A|\psi_0\rangle=a$ for a given $a\in\mathbb{R}$. Reimann \cite{Reimann2018a,Reimann2018b} argued that one also finds for most $\psi_0\in\mathbb{S}(\Hilbert)$ with $\langle\psi_0|A|\psi_0\rangle\approx a$ that $\langle\psi_t|B|\psi_t\rangle\approx b(t)$ for another observable $B$ and suitable ($\psi_0$-independent) $b(t)$. For technical reasons, these results do not cover the case that $A$ is a projection and $a=1$. In \cite{BRGSR18}, however, a quite general dynamical typicality result was proven that can also be applied to projections, and in \cite{TTV22-physik} we provided a simple proof for this special case. The result for projections can also be obtained as a special case of Eq.~(13) of Reimann \cite{Reimann07} (applied to a Gaussian distribution with covariance $P_\mu$, $K=1$, and $A=e^{iHt} P_\nu e^{-iHt}$).

As for generalized normal typicality, we only obtained bounds for the \emph{absolute} errors in \cite{TTV22-physik} (recapitulated as Theorem~\ref{thm: dyntyp phys} in Section~\ref{sec: prelim}).  Unfortunately, we cannot provide here an upper bound on the \emph{relative} errors either; but with the help of our improved version of the no-gaps delocalization result of Rudelson and Vershynin \cite{RV16} we are able to bound what we call the \textit{comparative error}, which is the absolute error divided by the time average, i.e.,
\begin{align}\label{comparativeerror}
    \frac{\bigl| \|P_\nu\psi_t\|^2 - \mathbb{E}_\mu\|P_\nu\psi_t\|^2 \bigr|}{\overline{\|P_\nu\psi_t\|^2}},
\end{align}
where $\EEE_\mu$ means the average over initial wave functions from the unit sphere in $\Hilbert_\mu$ and the overbar means time average over $[0,\infty)$. In Section~\ref{subsec: DynTyp}, we formulate a rigorous result about \eqref{comparativeerror} as Theorem \ref{thm: dyntyp rel 2}. It provides an upper bound on \eqref{comparativeerror} valid with high probability for random $H$ (distributed as before) for every $t$ and most $\psi_0\in\SSS(\Hilbert_\mu)$. If the constants, in particular $d_\mu,d_\nu$, and $D$, are such that this bound is small, then this means that $\|P_\nu\psi_t\|^2$ is nearly deterministic. Moreover, Theorem \ref{thm: dyntyp rel 2} says that also the whole function $t\mapsto \|P_\nu\psi_t\|^2$ is close to the function $t\mapsto \EEE_\mu\|P_\nu\psi_t\|^2$ on any interval $[0,T]$ in the $L^2$ norm.

\subsection{Contents}

The remainder of this paper is organized as follows: 
In Section \ref{sec: prelim}, we introduce some notation and recall our previous results from \cite{TTV22-physik}. In Section \ref{sec: main results}, we formulate our main results. In Section \ref{sec: no gaps}, we describe the result of Rudelson and Vershynin \cite{RV16} that we use (with minor corrections, see Theorem~\ref{thm: no gaps}) and point out which Gaussian random matrices will satisfy their hypotheses (Theorems~\ref{thm: GM bdd}, \ref{thm: GM nonzero mean}). In Section \ref{sec: more gen and proofs}, we prove our main results after formulating them in a somewhat more general version (Theorems~\ref{thm: main gen}, \ref{thm:11}) that applies to arbitrary operators $B$ instead of $P_\nu$. In Section \ref{sec: Improved LB}, we prove with the help of a delocalization result from Ajanki, Erdős, and Krüger \cite{AEK17} (quoted here as Theorem~\ref{thm: AEK deloc}) some improved lower bounds for the $M_{\mu\nu}$ that are nontrivial if $\mu$ or $\nu$ is the thermal equilibrium macro state (Theorem~\ref{thm: lb AEK}). In Section \ref{sec: ETH} we show that sharper estimates for all $M_{\mu\nu}$ follow from the version of the ETH due to Cipolloni, Erdős, and Henheik \cite{CEH23}. In Appendix~\ref{app: no res}, we include the proof that, with probability 1, a random $H$ with continuous distribution has no degeneracies or resonances. In Appendix \ref{sec: numerics}, we give more detail about our numerical examples.

\section{Prior Results} \label{sec: prelim}

We quote here the main results of \cite{TTV22-physik} about the absolute errors as Theorems~\ref{thm: abs errors A}, \ref{thm: GNT}, and~\ref{thm: dyntyp phys} alongside some related statements.
We start with a few definitions and in particular make precise some notions that appeared in the introduction: 

\begin{defn} \textbf{Most $\psi$, most $t$, and time-averages}.\\ 
Suppose that for each $\psi\in\mathbb{S}(\Hilbert)$, the statement $s(\psi)$ is either true or false, and let $\varepsilon>0$. We say that $s(\psi)$ is true for $(1-\varepsilon)$-\textit{most} $\psi\in\mathbb{S}(\Hilbert)$ if and only if
\begin{align}
    u\left(\left\{\psi\in\mathbb{S}(\Hilbert): s(\psi) \right\}\right) \geq 1-\varepsilon,
\end{align}
where $u$ is the normalized uniform measure over $\mathbb{S}(\Hilbert)$.\footnote{We do not consider here other ensembles of wave functions such as GAP \cite{gen-can} or the class of ensembles considered by Reimann in \cite{Reimann07}.}\\

Suppose that for each $t\in[0,\infty)$, the statement $s(t)$ is either true or false, and let $T,\delta>0$. We say that $s(t)$ is true for $(1-\delta)$-\textit{most} $t\in[0,T]$ if and only if 
\begin{align}
   \frac{1}{T}\lambda\left(\left\{t\in[0,T]: s(t)\right\} \right) \geq 1-\delta,
\end{align}
where $\lambda$ denotes the Lebesgue measure on $\mathbb{R}$. Moreover, we say that $s(t)$ is true for $(1-\delta)$-\textit{most} $t\in[0,\infty)$ if and only if
\begin{align}
    \liminf_{T\to\infty} \frac{1}{T}\lambda\left(\left\{t\in[0,T]: s(t)\right\} \right) \geq 1-\delta.
\end{align}
For any function $f: [0,\infty) \to \mathbb{C}$, define its \textit{time average} as
\begin{align}
    \overline{f(t)} := \lim_{T\to\infty} \frac{1}{T} \int_0^T f(t)\; dt
\end{align}
whenever the limit exists.
\end{defn}

In the following we consider Hamiltonians with spectral decomposition
\begin{align}
    H=\sum_{e\in\mathcal{E}} e\,\Pi_e,
\end{align}
where $\mathcal{E}$ is the set of distinct eigenvalues of $H$ and $\Pi_e$ the projection onto the eigenspace of $H$ with eigenvalue $e$.

\begin{defn}\textbf{Relevant properties of the Hamiltonian.}\\
Let $\kappa>0$. We define the \textit{maximum degeneracy} of an eigenvalue as
\begin{align}
    D_E := \max_{e\in\mathcal{E}} \tr(\Pi_e)
\end{align}
and the \textit{maximal number of gaps in an interval of length $\kappa$} as
\begin{align}
    G(\kappa) := \max_{E\in\mathbb{R}} \# \left\{(e,e')\in\mathcal{E}\times \mathcal{E}: e\neq e'\; \mbox{and}\; e-e' \in [E,E+\kappa) \right\}.
\end{align}
Moreover, we define the \textit{maximal gap degeneracy} as 
\be
D_G := \lim_{\kappa \to 0^+} G(\kappa).
\ee
\end{defn}

\begin{defn}\textbf{Asymptotic superposition weights.}\label{defn: sup weights}\\
Let $B\in\mathcal{L}(\Hilbert)$. Given any $\psi\in\mathbb{S}(\Hilbert)$ and any macro state $\mu$, define
\begin{align}
    M_{\psi B} &:= \sum_e \tr\left(|\psi\rangle\langle\psi|\Pi_e B \Pi_e \right),\\
    M_{\mu B} &:= \frac{1}{d_\mu} \sum_e \tr(P_\mu \Pi_e B \Pi_e).\label{eq: MmuB}
\end{align} In the special case that $B=P_\nu$ for some macro state $\nu$ we define 
\be
M_{\psi\nu} := M_{\psi P_\nu}
\ee
and
\be
M_{\mu\nu} := M_{\mu P_\nu}\,.
\ee
\end{defn}

In \cite{TTV22-physik} we proved the following theorem concerning the absolute errors:

\begin{thm}\label{thm: abs errors A} Let $B\in\mathcal{L}(\Hilbert)$, let $\varepsilon,\delta,\kappa,T >0$ and let $\mu$ be an arbitrary macro state. Then $(1-\varepsilon)$-most $\psi_0\in\mathbb{S}(\mathcal{H}_\mu)$ are such that for $(1-\delta)$-most $t\in [0,T]$
\begin{align}
    \bigl| \langle\psi_t|B|\psi_t\rangle - M_{\mu B}\bigr| \leq 4 \left(\frac{D_E G(\kappa) \|B\|}{\delta\varepsilon d_\mu} \left(1+\frac{8\log_2 d_E}{\kappa T}\right) \min\left\{\|B\|, \frac{\tr(|B|)}{d_\mu}\right\}\right)^{1/2},
\end{align}
where $D_E$ is the maximum degeneracy of an eigenvalue and $D_G$ is the maximum degeneracy of an eigenvalue gap of $H$.
\end{thm}

By setting $B=P_\nu$ for some macro state $\nu$, choosing $\kappa$ small enough such that $G(\kappa)=D_G$ and then taking the limit $T\to\infty$ we immediately obtain:

\begin{thm}[Generalized normal typicality: absolute errors]\label{thm: GNT} Let $\varepsilon,\delta>0$ and let $\mu,\nu$ be two arbitrary macro states. Then $(1-\varepsilon)$-most $\psi_0\in\mathbb{S}(\Hilbert_\mu)$ are such that for $(1-\delta)$-most $t\in [0,\infty)$ 
\begin{align}
    \bigl|\|P_\nu\psi_t\|^2 - M_{\mu\nu}\bigr| \leq 4\left(\frac{D_E D_G}{\delta\varepsilon d_\mu}\min\left\{1,\frac{d_\nu}{d_\mu}\right\}\right)^{1/2}.\label{ineq: GNT}
\end{align}
\end{thm}

Theorem \ref{thm: GNT} tells us, roughly speaking, that as soon as
\begin{align}
    d_\mu \gg D_E D_G,
\end{align}
i.e., as soon as the dimension of $\Hilbert_\mu$ is huge and no eigenvalues or gaps are macroscopically degenerate, the superposition weight $\|P_\nu\psi_t\|^2$ will be close to the fixed value $M_{\mu\nu}$ for most times $t$ and most initial states $\psi_0\in\mathbb{S}(\Hilbert_\mu)$.

Of course, the fixed value $M_{\mu\nu}$, or more generally $M_{\mu B}$, can be computed by taking the average of over time and over the sphere in $\Hilbert_\mu$. This is the content of the following well known proposition, see also \cite{TTV22-physik}.

\begin{prop}
Let $B\in \mathcal{L}(\Hilbert)$. Then
\begin{align}
    M_{\psi_0 B} &= \overline{\langle\psi_t|B|\psi_t\rangle}\; \mbox{and} \label{MpsiBtimeavg}\\
    M_{\mu B} &= \int_{\mathbb{S}(\Hilbert_\mu)} M_{\psi_0 B}\; u_\mu(d\psi_0) = \EEE_\mu M_{\psi_0 B},
\end{align}
where $u_\mu$ is the normalized uniform measure over $\mathbb{S}(\Hilbert_\mu)$.
\end{prop}
This motivates us to call $M_{\psi\nu}$ the \textit{time average superposition weight} in $\Hilbert_\nu$ and $M_{\mu\nu}$ the \textit{full average superposition weight} in $\Hilbert_\nu$.

For a system with $N$ particles or more generally $N$ degrees of freedom the dimension $D$ is of order $\exp(N)$. For any macro state $\mu$ we define $s_\mu$, the entropy per particle in the macro state $\mu$, by
\begin{align}
    d_\mu =: \exp(s_\mu N/k_B),
    \label{eq: dimensions smu}
\end{align}
where $k_B$ is the Boltzmann constant.

With \eqref{eq: dimensions smu} for the dimensions of the macro spaces, we find the following corollary of Theorem \ref{thm: GNT} whose proof can also be found in \cite{TTV22-physik}: 

\begin{cor}\label{cor: GNT}
Assume \eqref{eq: dimensions smu} and let $\varepsilon,\delta>0$. Then, for all macro states $\mu, \nu_-, \nu_+$ with
\begin{align}
    s_{\nu_-}\leq s_{\mu} \leq s_{\nu_+}
\end{align}
it holds for $(1-\varepsilon)$-most $\psi_0\in\mathbb{S}(\Hilbert_\mu)$ for $(1-\delta)$-most $t\in[0,\infty)$
\begin{align}
    \bigl|\|P_{\nu_+}\psi_t\|^2 - M_{\mu\nu_+} \bigr| &\leq \frac{4\sqrt{D_E D_G}}{\sqrt{\varepsilon\delta}} \exp\left(-\frac{s_\mu N}{2k_B}\right),\\
    \bigl|\|P_{\nu_-}\psi_t\|^2 - M_{\mu\nu_-} \bigr| &\leq \frac{4\sqrt{D_E D_G}}{\sqrt{\varepsilon\delta}} \exp\left(-\frac{\left(s_\mu-\frac{s_{\nu_-}}{2}\right) N}{k_B}\right).
\end{align}
\end{cor}

Corollary \ref{cor: GNT} shows that fluctuations of the time-dependent superposition weights around their expected values are exponentially small in the number of particles provided that no eigenvalues and gaps are macroscopically degenerate. 

In addition to the theorem concerning generalized normal typicality, we also proved a result concerning dynamical typicality in \cite{TTV22-physik}:

\begin{thm}\label{thm: dyntyp phys}
Let $\mu$ be an arbitrary macro state and let $B$ be any operator on $\Hilbert$. Let $w_{\mu B}:\mathbb{R}\to[0,1]$ be the function defined by
\begin{align}
     w_{\mu B}(t) := \frac{1}{d_\mu}\tr\left[P_\mu\exp(iHt) B\exp(-iHt)\right]\label{eq: wmuB}
\end{align}
Then, for every $t\in\mathbb{R}$ and every $\varepsilon>0$, $(1-\varepsilon)$-most $\psi_0\in\mathbb{S}(\Hilbert_\mu)$ are such that
\begin{align}
\bigl|\langle\psi_t|B|\psi_t\rangle-w_{\mu B}(t)\bigr| \leq \min\left\{\frac{\|B\|}{\sqrt{\varepsilon d_\mu}}, \sqrt{\frac{\|B\|\tr(|B|)}{\varepsilon d_\mu^2}},\sqrt{\frac{18\pi^3\log(4/\varepsilon)}{d_\mu}}\|B\|\right\}.\label{ineq: dyntyp1}
\end{align}
Moreover, for every $\mu$ and $B$, every $T>0$, and $(1-\varepsilon)$-most $\psi_0\in\mathbb{S}(\Hilbert_\mu)$,
\begin{align}
    \frac{1}{T}\int_0^T \bigl|\langle\psi_t|B|\psi_t\rangle - w_{\mu B}(t)\bigr|^2\, dt \leq \frac{\|B\|^2}{\varepsilon d_\mu}.\label{ineq: dyntyp2}
\end{align}
\end{thm}

We note first that $w_{\mu B}(t)$ is exactly the ensemble average of $\langle \psi_t|B|\psi_t\rangle$, i.e., the average over $\psi_0\in\SSS(\Hilbert_\mu)$:
\be
w_{\mu B}(t) = \EEE_\mu \scp{\psi_t}{B|\psi_t}\,.
\ee
In particular, by Fubini's theorem,
\be
\overline{w_{\mu B}(t)} = M_{\mu B}\,.
\ee
Eq.~\eqref{ineq: dyntyp1} offers three estimates that can be useful depending on the sizes of $B$, $\varepsilon$, and $d_\mu$. Specifically, if $\|B\|$ is of order 1 and $d_\mu\gg 1/\varepsilon$, then the first estimate implies that $\langle\psi_t|B|\psi_t\rangle$ is close to its average $w_{\mu B}(t)$. We note further that the first two estimates in \eqref{ineq: dyntyp1} can be obtained by similar arguments as in the proof of Theorem \ref{thm: abs errors A} whereas the third estimate is a consequence of Lévy's lemma (see e.g. \cite[Sec. II.C]{SWGW22}). 

In the following we can assume without loss of generality $D_E=D_G=1$ since we consider random matrices whose joint distribution of their entries is absolutely continuous with respect to the Lebesgue measure. The set of matrices with degenerate eigenvalues or eigenvalue gaps -- we also say that the matrix has \textit{resonances} -- has Lebesgue measure zero and therefore, with probability 1, $D_E=D_G=1$. For the convenience of the reader, we give the proof of this statement in Appendix \ref{app: no res}.

\section{Main Results\label{sec: main results}}

In this section we present and discuss our main results concerning lower bounds for the $M_{\mu\nu}$ and therefore generalized normal typicality (Theorem \ref{thm: main new}) as well as a corollary thereof bounding the relative errors (Corollary \ref{cor: real dim rel}) and a result bounding the comparative error for dynamical typicality (Theorem \ref{thm: dyntyp rel 2}). In all the results presented in this section it is assumed that the Hamiltonian is of a special form and we only consider projections $P_\nu$ onto a macro space $\Hilbert_\nu$. More general results and the proofs can be found in Section~\ref{sec: more gen and proofs}.

\subsection{Generalized Normal Typicality}

We make the following assumption:

\begin{asm}\label{asm: Hamiltonian}
The Hamiltonian $H$ is of the form $H=H_0+V$, where $H_0$ is a (deterministic) Hermitian $D\times D$ matrix and $V$ is a Hermitian random Gaussian perturbation, more precisely, $V=\frac{1}{\sqrt{2}}(A+A^{*})$, where $A=(a_{ij})$ is a random $D\times D$ matrix with independent Gaussian entries with mean zero, i.e., all random variables $\RE a_{ij},\IM a_{ij}, i,j\in[D]$, are independent and $\RE a_{ij},\IM a_{ij} \sim \mathcal{N}(0,\sigma_{ij}^2/2)$ for some $\sigma_{ij}>0$ that fulfill $\sigma_{ij}=\sigma_{ji}$.
\end{asm}

Note that because of Assumption \ref{asm: Hamiltonian} the matrix $V$ is Hermitian and Gaussian, more precisely, $\RE v_{ij},\IM v_{ij} \sim \mathcal{N}(0,\sigma_{ij}^2)$ for $i\neq j$ and $v_{ii}=\RE v_{ii} \sim \mathcal{N}(0,\sigma_{ii}^2)$, i.e., $V$ is a Hermitian Gaussian Wigner-type matrix.

For a Hamiltonian as in Assumption \ref{asm: Hamiltonian} we define
\begin{align}
    C_{H_0} := D^{-1/2} \max\{\|\RE H_0\|, \|\IM H_0\|\}.\label{const: CH0}
\end{align}
Note that for a typical many-body Hamiltonian we expect that $\|\RE H_0\|, \|\IM H_0\| \sim \log D$ and therefore $C_{H_0}$ should be very small for large $D$.

\begin{thm}[Lower bounds for $M_{\mu\nu}$]\label{thm: main new}
Let $\varepsilon' \in (0,\frac{1}{2})$ and let $\mu$ and $\nu$ be arbitrary macro states such that $d_\mu,d_\nu>\max\{166,4|\log_2(\varepsilon'/\sqrt{2})|\}$. Let $H$ satisfy Assumption \ref{asm: Hamiltonian}.
Let $\sigma_- := \min_{i,j} \sigma_{ij}$ and $\sigma_+ := \max_{i,j}\sigma_{ij}$. Then with probability at least $1-\varepsilon'$,
\begin{align}
    M_{\mu\nu} \geq \left(\sqrt{\varepsilon' C_{\sigma}}\,\frac{\max\left\{d_\mu,d_\nu\right\}}{D}\right)^{16} \min\left\{1,\frac{d_\nu}{d_\mu}\right\}, \label{mainnewMmunu}
\end{align}
where 
\begin{align}
C_{\sigma}:=\frac{c_- \sigma_-}{c_+\sigma_++C_{H_0}}\label{const Csigma}
\end{align}
with $C_{H_0}$ defined in \eqref{const: CH0} and absolute constants $c_-,c_+>0$.
\end{thm}

The main application of this theorem is provided by our next result: the analogue of Corollary \ref{cor: GNT} for the relative errors. 
In addition to \eqref{eq: dimensions smu}, define   $s_{\textup{mc}}$ by
\begin{align}
    D = \exp\left(s_{\textup{mc}} N /k_B\right).
\end{align}
Since we expect that the dimension of the thermal equilibrium macro space $d_{\textup{eq}}$ is extremely large, it should hold that $s_{\textup{mc}}\approx s_{\textup{eq}}$.

\begin{cor}[Generalized normal typicality -- relative errors]\label{cor: real dim rel}
Let $\varepsilon,\delta >0$, $\varepsilon'\in~(0,\frac{1}{2})$ and let $\mu$ and $\nu$ be macro states such that $d_\mu,d_{\nu}>\max\{166,4|\log_2(\varepsilon'/\sqrt{2})|\}$.
Let $H$ be a random Hermitian $D\times D$ matrix as in Theorem~\ref{thm: main new}. Then with probability at least $1-\varepsilon'$, $(1-\varepsilon)$-most $\psi_0\in\mathbb{S}(\Hilbert_\mu)$ are such that for $(1-\delta)$-most $t\in[0,\infty)$,
\begin{align}
 \frac{\bigl| \|P_{\nu}\psi_t\|^2-M_{\mu\nu} \bigr|}{M_{\mu\nu}} \leq \frac{4}{\sqrt{\varepsilon\delta}} (C_{\sigma}\varepsilon')^{-8}\exp\left(-\frac{N}{2k_B}\left(\min\{s_\mu,s_\nu\}-32(s_{\textup{mc}}-\max\{s_\mu,s_\nu\})\right)\right),
\end{align}
where $C_{\sigma}$ is defined in \eqref{const Csigma}.
\end{cor}

The most important aspects of this bound are that it bounds the \emph{relative} error, it shrinks exponentially as $N$ increases, and the rate of shrinking can be given explicitly.

\begin{rmk}\label{rmk: rel error small}
 It is sometimes desirable to consider a sequence of systems with increasing dimension $D\to\infty$, corresponding for example to an increasing particle number $N\to\infty$, for which it makes sense to talk of the ``same'' macro states $\nu$ for each member of the sequence. In that case, suppose that $C_\sigma$ is bounded below away from zero uniformly in $D$ (as  $C_{H_0}$ is small in typical models, this is basically a condition on $\sigma_-$ and $\sigma_+$). Then the relative errors in Corollary~\ref{cor: real dim rel} are small if $s_{\textup{mc}} < \max\{s_\mu,s_\nu\}+\min\{s_\mu,s_\nu\}/32$. Since we expect that $s_{\textup{mc}}\approx s_{\textup{eq}}$, one of the macro spaces $\Hilbert_\nu$ or $\Hilbert_\mu$ thus  needs to have specific entropy not too far from the one of the equilibrium macro state. While this might seem quite restrictive at first sight, observe that even if we assume that the $M_{\mu\nu}$ scale like in the case of normal typicality, i.e.,
\begin{align}
     M_{\mu\nu} \approx \frac{d_\nu}{D} \approx \exp\left(-\frac{s_{\textup{mc}}-s_\nu}{k_B}N\right),
 \end{align}
 then the relative errors are small only if $s_{\textup{mc}}<\max\{s_\nu,s_\mu\}+\min\{s_\nu,s_\mu\}/2$ (recall Theorem \ref{thm: GNT} for the absolute errors), see also the discussion in \cite{TTV22-physik}. 
\end{rmk}

\begin{rmk}
    Suppose again that $C_{H_0}$ is bounded uniformly in $D$ and that $\sigma_\pm = D^{\alpha_\pm}$ with $\alpha_- \leq \alpha_+$. In this case the upper bound for the relative errors becomes
    \begin{multline}
        \frac{\bigl|\|P_\nu\psi_t\|^2-M_{\mu\nu} \bigr|}{M_{\mu\nu}} \leq \frac{4}{\sqrt{\varepsilon\delta}(c_-\varepsilon')^8} \left(c_+ \exp(\alpha_+ s_{\mathrm{mc}}N/k_B)+C_{H_0}\right)^8\:\times\\
        \times\:\exp\left(-\frac{N}{2k_B}\left(\min\{s_\mu,s_\nu\}-32\left(\left(1-\frac{\alpha_-}{2}\right)s_{\textup{mc}}-\max\{s_\mu,s_\nu\}\right)\right)\right).\label{ineq: up bound rmk 2}
    \end{multline}
Suppose first that $\alpha_+<0$. Then, for large $N$, the second factor can be bounded by $(2C_{H_0})^8$ and the right-hand side is small if the expression in the exponential function is negative, i.e., if
\begin{align}
    \left(1-\frac{\alpha_-}{2}\right)s_{\mathrm{mc}}< \max\{s_\mu,s_\nu\} + \frac{1}{32}\min\{s_\mu,s_\nu\}.\label{ineq: cond err small alpha}
\end{align}
Since $\alpha_-<0$ the macro state $\mu$ or $\nu$ has to be even closer to the equilibrium macro state than in the case that the variances $\sigma_-,\sigma_+$ do not depend on the dimension $D$. 
If $\alpha_+=0$ then the second factor in \eqref{ineq: up bound rmk 2} can also be bounded by some constant and we obtain again \eqref{ineq: cond err small alpha} as a condition for the relative errors to be small. If, however, $\alpha_+ > 0$, the second factor can be bounded by a quantity proportional to $\exp(8\alpha_+ s_{\mathrm{mc}} N/k_B)$ and therefore the right-hand side of \eqref{ineq: up bound rmk 2} is small if
\begin{align}
    \left(1+\frac{\alpha_+ -\alpha_-}{2}\right) s_{\mathrm{mc}} < \max\{s_\mu,s_\nu\} + \frac{1}{32}\min\{s_\mu,s_\nu\}.
\end{align}
Thus, if $\alpha_+>\alpha_-$, then the macro state $\mu$ or $\nu$ again has to be closer to the equilibrium macro state than in the case discussed in Remark \ref{rmk: rel error small}. Note that, of course, we can also assume that $\sigma_\pm = c_{\sigma_\pm}D^{\alpha_\pm}$ with constants $c_{\sigma_\pm}>0$ and the conditions under which the relative errors are small remain the same since the upper bounds only change by multiplicative constants. 
\end{rmk}

\subsection{Dynamical Typicality\label{subsec: DynTyp}}

We now turn to our main result about dynamical typicality. We would like to give again a bound on the \emph{relative error}, but this we can only provide for \emph{most} (rather than \emph{all}) times. But for \emph{most} times, we already know from generalized normal typicality that $\|P_\nu\psi_t\|^2$ is close to $M_{\mu\nu}$ (as visible in Figure \ref{fig:numerical example 2}). On the other hand, for dynamical typicality, we are particularly interested in times before normal equilibration has set in (as in Figure \ref{fig:numerical example 1}). For this situation, we can provide a bound on what we call the \emph{comparative error} and that is given by the quotient of the absolute error (which in this case is time-dependent) and the time average of the quantity considered. In this way, we compare the absolute error to a comparison value that expresses in a way what magnitude to expect of $\|P_\nu \psi_t\|^2$. Smallness of the comparative error at time $t$ means that the absolute error at $t$ is much smaller than the long-time value of $\|P_\nu \psi_t\|^2$, although not necessarily much smaller than the instantaneous ensemble average
\be
\EEE_\mu \|P_\nu\psi_t\|^2=w_{\mu P_\nu}(t) =: w_{\mu\nu}(t)\,.
\ee
This statement still gives us a handle on comparatively small $\Hilbert_\nu$ for which $\|P_\nu\psi_t\|^2$ is small in absolute terms for most $t$. Now recall from \eqref{MpsiBtimeavg} that
\be
M_{\psi_0 \nu}=\overline{\|P_\nu\psi_t\|^2}\,.
\ee

\begin{thm}[Dynamical typicality -- comparative errors]\label{thm: dyntyp rel 2}
Let $\varepsilon>0$, $\varepsilon'\in(0,\frac{1}{2})$ and let $\mu$ and $\nu$ be macro states such that $d_\nu> \max\{166, 4|\log_2\varepsilon'|\}$.
Let $H$ be a random Hermitian $D\times D$ matrix as in Theorem~\ref{thm: main new}. Then with probability at least $1-\varepsilon'$, for each $t\in\mathbb{R}$, $(1-\varepsilon)$-most $\psi_0\in\mathbb{S}(\Hilbert_\mu)$ are such that 
\begin{align}\label{comparative}
   \frac{\bigl|\|P_\nu\psi_t\|^2- w_{\mu\nu}(t)\bigr|}{M_{\psi_0 \nu}} &\leq \frac{1}{\sqrt{\varepsilon}} (C_\sigma \varepsilon')^{-8} \exp\left(-\frac{N}{2k_B}\left(2s_\mu-\min\{s_\mu,s_\nu\}-32(s_{\textup{mc}}-s_\nu)\right)\right)
\end{align}
where $C_{\sigma}$ is defined in \eqref{const Csigma}. Moreover, with probability $1-\varepsilon'$, for every $T>0$, $(1-\varepsilon)$-most $\psi_0\in\mathbb{S}(\Hilbert_\mu)$ are such that
\begin{align}
    \frac{1}{T}\int_0^T \frac{\bigl| \|P_\nu\psi_t\|^2- w_{\mu\nu}(t) \bigr|^2}{M_{\psi_0 \nu}^2}\, dt \leq \frac{1}{\varepsilon} (C_{\sigma}\varepsilon')^{-16} \exp\left(-\frac{N}{k_B}\left(s_\mu-32(s_{\textup{mc}}-s_\nu)\right)\right).\label{ineq: dyntyp rel new 3}
\end{align}
\end{thm}

Again, key aspects are that the upper bounds in \eqref{comparative} and \eqref{ineq: dyntyp rel new 3} shrink exponentially with increasing $N$ at explicit rates. The bounds
are small if $s_{\textup{mc}}<s_\nu+s_\mu/16-\min\{s_\mu,s_\nu\}/32$ resp. $s_{\textup{mc}}<s_\nu+s_\mu/32$, i.e., $s_\nu$ must be close to $s_{\textup{mc}}$ and therefore close to $s_{\textup{eq}}$. 
In that case, \eqref{ineq: dyntyp rel new 3} shows that the curve $t\mapsto \|P_\nu\psi_t\|^2$ is close to the curve $t\mapsto\mathbb{E}_\mu\|P_\nu\psi_t\|^2$ in the $L^2$-sense even when compared to the (possibly small) time average $\overline{\|P_\nu\psi_t\|^2}$.

\section{No-Gaps Delocalization and Band Matrices \label{sec: no gaps}}

Our main results are based on a ``no-gaps delocalization'' result of Rudelson and Vershynin (2016) \cite[Theorem 1.3]{RV16}. In this section, we state this result (with certain corrections, Theorem~\ref{thm: no gaps} in Section \ref{subsec: RV corrected}), provide an extension of it that we will use when proving our main results (Theorem~\ref{thm: no gaps extended} in Section \ref{subsec: RV extended}), and provide examples of random matrices for which the hypotheses of Rudelson and Vershynin are satisfied (Theorems~\ref{thm: GM bdd} and \ref{thm: GM nonzero mean} in Section \ref{subsec: ex}).

\subsection{Statement of the No-Gaps Delocalization Result of Rudelson and Vershynin}
\label{subsec: RV corrected}

Since explicit exponents and the dependence of constants on the parameters play an important role for our purposes, we have followed the exponents and constants through the proof of Rudelson and Vershynin in \cite{RV16}. It turned out that in some places we arrived at different values than stated in \cite{RV16}. In this subsection, we state the relevant result with these values, and provide a discussion of how we arrived at them.

In the following we use the following abbreviation.
For a random $D\times D$ matrix $H$ and any number $J>0$ we introduce the ``boundedness event''
\begin{align}
    \mathcal{B}_{H,J} := \left\{\|H\| \leq J\sqrt{D}\right\}.
\end{align}
It will turn out in Section~\ref{subsec: ex} that for relevant examples of random matrices $H$ and large $D$, the event $\mathcal{B}_{H,J}$ occurs with high probability; this conclusion makes use of a result of Latala~\cite{Latala05}.

So here is our adjusted statement of Theorem 1.3 of \cite{RV16}.

\begin{thm}[No-gaps delocalization; Rudelson, Vershynin (2016)]\label{thm: no gaps}
Let $H=(h_{ij})$ be a $D\times D$ random matrix such that for $i,j\in [D]$ the (continuous) random variable $\RE h_{ij}$ is independent of the other entries of $\RE H$ except possibly $\RE h_{ji}$, the densities $\varrho_{ij}^{\RE}$ of $\RE h_{ij}$ are bounded by some number $K\geq 1$ and the imaginary part is fixed. Choose $J\geq 1$ such that the boundedness event $\mathcal{B}_{H,J}$ holds with probability at least $1/2$. Let $\kappa \in (180/D,1/2)$ and $s>0$. Then, conditionally on $\mathcal{B}_{H,J}$, the following holds with probability at least $1-(Cs)^{\kappa D}$: Every eigenvector $v$ of $H$ satisfies
\begin{align}
    \|v_I\| \geq (\kappa s)^9 \|v\| \quad \mbox{for all}\; I \subset [D]  \mbox{ with } |I|\geq \kappa D,\label{ineq: nogaps deloc}
\end{align}
where
\begin{align}
    \|v_I\| := \left(\sum_{j\in I}|v_j|^2 \right)^{1/2}
\end{align}
and
$C=C(K,J)\geq 1$.
\end{thm}

Here is how this statement differs from Theorem 1.3 in \cite{RV16}. The lower bound on $\varepsilon$ is changed from $8/D$ to $180/D$ and the exponent of $\kappa s$ in \eqref{ineq: nogaps deloc} from 6 to 9. Therefore the theorem tells us that subsets $I\subset [D]$ of at least 180 (instead of 8) coordinates of any eigenvector carry a non-negligible part of its mass. Moreover, the lower bound obtained for $\|v_I\|$ is slightly worse since $\kappa<1$ by assumption and we can assume without loss of generality that also $s<1$ because otherwise the lower bound on the probability, $1-(Cs)^{\kappa D}$, is negative.

Here is how we arrived at these exponents and constants.
In the rest of this section, all numbers of theorems, subsections etc. refer to the ones in \cite{RV16} if not stated otherwise. We start from Theorem 5.1 (Section 5.1) to derive Theorem 1.3. 
Let $C_2$ be the constant $C$ in Theorem 5.1 and define for $\delta,\kappa>0$ the \textit{gap event} (called the \textit{localization event} in \cite{RV16})
\begin{align}
   \mbox{Gap}(H,\kappa,\delta) := \left\{\exists \; \mbox{eigenvector}\; v \in \mathbb{C}^D, \|v\| = 1, \exists I\subset [D], |I|=\kappa D: \|v_I\|<\delta \right\}\,,
\end{align}
i.e., the event that in an eigenvector $v$ of $H$ a fraction $\kappa$ of basis vectors is underrepresented in  $v$.

After an application of Proposition 4.1 one arrives at
\begin{align}
    \mathbb{P}\left(\mbox{Gap}(H,\kappa,t/8J)\; \mbox{and}\; \mathcal{B}_{H,J}\right) \leq 5\left(\frac{8J}{t}\right)^2 (e/\kappa)^{\kappa D} p_0 \label{ineq: Loc}
\end{align}
where $t\geq 0$ and $p_0 = (2C_2 KJ t^{0.4}\kappa^{-1.4})^{\kappa D/2}$. Set $t=8J (\kappa s)^{\alpha}$ for some $\alpha>0$. (In \cite{RV16}, $\alpha=6$ was stated, while we arrive at $\alpha=9$ below, so let us keep the value unspecified for a little while.) What we want to show is that
\be
\PPP\Bigl(\mbox{Gap}(H,\kappa,(\kappa s)^{\alpha})\Big| \mathcal{B}_{H,J}\Bigr) \leq (Cs)^{\kappa D}\,,
\ee
where $C$ can depend on $J$ and $K$ but not on $\kappa,s$, or $D$. Since $\PPP(\mathcal{B}_{H,J})\geq 1/2$, it suffices to bound $\PPP(\mbox{Gap}(H,\kappa,(\kappa s)^{\alpha})\cap \mathcal{B}_{H,J})$ in a similar way. From \eqref{ineq: Loc}
we obtain that
\begin{subequations}
\begin{align}
    &\mathbb{P}\left(\mbox{Gap}(H,\kappa,(\kappa s)^{\alpha})\; \mbox{and}\; \mathcal{B}_{H,J}\right) \nonumber\\[2mm]
    &\qquad\leq 5(\kappa s)^{-2\alpha} (e/\kappa)^{\kappa D} \left(2C_2 KJ t^{0.4} \kappa^{-1.4}\right)^{\kappa D/2}\\
    &\qquad=\left(5^{1/\kappa D} \kappa^{-2\alpha/\kappa D+0.2\alpha-1.7} s^{-2\alpha/\kappa D+0.2\alpha-1} e \sqrt{2C_2 KJ}(8J)^{0.2}s\right)^{\kappa D}\\
    &\qquad\leq \left(\kappa^{-2\alpha/\kappa D+0.2\alpha-1.7} s^{-2\alpha/\kappa D+0.2\alpha-1} 5e\sqrt{2C_2 KJ}(8J)^{0.2}s\right)^{\kappa D}.
\end{align}
\end{subequations}
So we would like a constant (independent of $\kappa,s,D$) as an upper bound for 
\be\label{kappapowerspower}
\kappa^{-2\alpha/\kappa D+0.2\alpha-1.7} s^{-2\alpha/\kappa D + 0.2\alpha-1},
\ee
where $\kappa<1/2$ and $s<1$. However, in the limit $\kappa D\to 0$, it follows that $-2\alpha/\kappa D\to -\infty$ and $\kappa^{-2\alpha/\kappa D}\to \infty$, so we need to exclude $\kappa$ values from a neighborhood of 0 and require, say, $\kappa\geq \ell/D$ for some $\ell\in\mathbb{N}$ (as explicitly done in \cite{RV16} and in Theorem~\ref{thm: no gaps}). But then $\kappa$ can still be quite close to 0, and its exponent in \eqref{kappapowerspower} should be $\geq 0$ or else \eqref{kappapowerspower} will not be bounded as $D\to\infty$, so we want that
\begin{align}
    -\frac{2\alpha}{\kappa D}+0.2\alpha-1.7 \geq 0 \,.
\end{align}
In particular, for $\alpha=6$ as in \cite{RV16}, this exponent would always be negative, and \eqref{kappapowerspower} for $\kappa=\ell/D$ would not be bounded. However, if we suppose that 
$\ell>10$ and $\alpha \geq \frac{17\ell}{2(\ell-10)}$, then, for every $\kappa\geq \ell/D$,
\begin{subequations}
\begin{align}
    -\frac{2\alpha}{\kappa D}+\frac{\alpha}{5}-\frac{17}{10} &\geq -\frac{2\alpha}{\ell}+\frac{\alpha}{5}-\frac{17}{10}\\
    &= \alpha \frac{\ell-10}{5\ell}-\frac{17}{10}\\
    &\geq 0.
\end{align}
\end{subequations}
In this case we obtain that
\begin{align}
    \mathbb{P}\left(\mbox{Gap}\left(H,\kappa,(\kappa s)^{\alpha}\right)\; \mbox{and}\; \mathcal{B}_{H,J}\right) &\leq \left(5e\sqrt{2C_2 KN}(8J)^{0.2} s\right)^{\kappa D}\\
    &=: \left(\tilde{C}s\right)^{\kappa D}
\end{align}
with $\tilde{C}= \tilde{C}(K,J)= 5e\sqrt{2C_2 KJ}(8J)^{0.2}$ being the desired constant.\footnote{More precisely, $C=2\tilde{C}$, where $C$ is the constant in Theorem \ref{thm: no gaps}. This is due to the fact that $\mathbb{P}\left(\mbox{Gap}(H,\kappa,(\kappa s)^\alpha)\; \mbox{and}\; \mathcal{B}_{H,J}\right)$ and not $\mathbb{P}\left(\mbox{Gap}(H,\kappa,(\kappa s)^\alpha)|\mathcal{B}_{H,J}\right)$ is considered in the proof of Theorem 1.3.} Note that $\ell\mapsto \frac{17\ell}{2(\ell-10)}$ for $\ell>10$ is monotonically decreasing and $\lim_{\ell\to\infty} \frac{17\ell}{2(\ell-10)}=8.5$. Therefore we can choose $\alpha=9$ and $\ell=180$ to get $-2\alpha/\kappa D+0.2\alpha-1.7\geq 0$.

\begin{rmk}[Constant in Theorem \ref{thm: no gaps}]\label{rmk: constant}
Since the constant $C(K,J)$ in 
Theorem \ref{thm: no gaps} will appear in the upper bound for the relative errors and we want to allow the case that $K$ and $J$ depend on $D$, we are interested in the dependence of $C$ on these two parameters. In \cite{RV16} it is said that all appearing constants (denoted by $C, c,...$) might depend on $K$ and $J$ and therefore we carefully investigated the constants in the theorems, propositions, lemmas etc.\ needed for the proof of their main theorem (here Theorem \ref{thm: no gaps}). To avoid confusion with the constant $C$ in the main theorem, we denote the other constants that are also denoted as $C$ in \cite{RV16} as $C_1,C_2,...$ in the order of their appearance; other constants that appear multiple times with a different value as well are also numbered in the following in the order of their appearance. An inspection of the proofs reveals that
\begin{itemize}
    \item $C_1 = e$ in Lemma 3.3,
    \item $C=2\tilde{C}$, where $\tilde{C}$ is defined above,
    \item $C_2 = 2 \max\left\{\left(\frac{4}{\tilde{c}_1}\right)^{0.9}C_6, C_4\sqrt{\tilde{c}_2}\right\}$ in Theorem 5.1, where $\tilde{c}_2$ is an absolute constant that bounds the constants $C(p)$ in (5.7) for $p\geq 2$, see also Theorem 3.21 in \cite{SW71}, and $\tilde{c}_1\in (0,0.2)$ is another absolute constant that can be chosen arbitrarily, see also Section 5.4.2 and 5.5.3 (note that in (5.9) a square is missing); the values of the other constants are given below,
    \item $C_3 = \frac{1}{4} C_2$ in (5.1), 
    \item $C_4 = \sqrt{2} \hat{C} \bar{C}$ in Lemma 5.4, where $\hat{C}, \bar{C}>0$ are absolute constants; more precisely, $\hat{C}$ appears in the upper bound for the density of a random vector $PX$ where $X=(X_1,\dots,X_D)$ is a vector of real-valued independent random variables whose densities are bounded by some constant a.e. and $P$ is an orthogonal projection onto a subspace of dimension $d$, see also \cite{RV15b} for details; moreover, the constant $\bar{C}$ appears in the upper bound of the volume of a $d$-dimensional ball of radius $\tau\sqrt{d}$ for some $\tau>0$\footnote{The constant $\bar{C}$ can be chosen as $\bar{C}=\sqrt{2\pi e}$ which can be easily seen as follows: For $\tau>0$, the volume of a $d$-dimensional ball with radius $\tau \sqrt{d}$ is given by
    \begin{align}
    \frac{\pi^{d/2}}{\Gamma\left(\frac{d}{2}+1\right)}\left(\tau\sqrt{d}\right)^d \leq \left(\frac{e}{\frac{d}{2}+1}\right)^{d/2}\left(\sqrt{\pi} \tau \sqrt{d}\right)^d\leq \left(\sqrt{2\pi e} \tau\right)^d = \left(\bar{C}\tau\right)^d
    \end{align}
    },
    \item $C_{0,1} = 4$ in Lemma 5.5,
    \item $C_{0,2} = 4e$ in Lemma 5.6,
    \item $C_5 = 72e$ in Lemma 5.7,
    \item $C_6 = 72e$ in Lemma 5.8.
\end{itemize}
Thus, we see that the constant $C_2$ depends neither on $K$ nor on $J$, and we conclude from our computation above that $C \sim \sqrt{K} J^{0.7}$.
\end{rmk} 

\subsection{An Extension\label{subsec: RV extended}}

In Theorem \ref{thm: no gaps}, the imaginary part of the random matrix is fixed, and it is assumed that $J\geq 1$. In this section we present and prove a theorem that covers matrices with random imaginary part as well, that also allows $0<J\leq 1$ and we improve the exponent of $\kappa s$ in the lower bound for $\|v_I\|$ from 9 to 8. More precisely, we prove the following theorem:

\begin{thm}\label{thm: no gaps extended}
Let $H=(h_{ij})$ be a $D\times D$ random matrix such that $\RE H$ and $\IM H$ are independent, for $i,j\in [D]$ the (continuous) random variable $\RE h_{ij}$ is independent of the other entries of $\RE H$ except possibly $\RE h_{ji}$, for $i,j\in [D], i\neq j$ the (continuous) random variable $\IM h_{ij}$ is independent of the other entries of $\IM H$ except possibly $\IM h_{ji}$, and the densities $\varrho_{ij}^{\RE}$ are bounded by some number $K>0$. Choose $J>0$ such that $JK\geq 1$ and the boundedness events $\mathcal{B}_{\RE H,J}$ and $\mathcal{B}_{\IM H,J}$ hold with probability at least $1-\eta$ for some $0<\eta\leq 1/2$. Let $\kappa \in (83/D,1/2)$ and $0<s\leq 1$. Then the following holds with probability at least $\left(1-(c_c\sqrt{KJ}s)^{\kappa D}\right)(1-\eta)^4$, where $c_c\geq 1$ is a universal constant: Every eigenvector $v$ of $H$ satisfies
\begin{align}
    \|v_I\| \geq (\kappa s)^8 \|v\| \quad \mbox{for all}\; I \subset [D]  \mbox{ with } |I|\geq \kappa D\,.
\end{align}
\end{thm}

\begin{proof}
We prove the theorem in three steps. The first step is to establish a modification of Theorem \ref{thm: no gaps} where the imaginary part of the matrix $H$ can also be random. The second step is to relax the condition $J\geq 1$ to $J>0$ (in this case we have to additionally assume that $JK\geq 1$), which can be done via a simple scaling argument. Finally, in the third step, we show that by a small modification in the proof of Rudelson and Vershynin, the lower bound for $\|v_I\|_2$ can be slightly improved.
\\
\\
\underline{1. Step (the complex case)}:
We first assume that $J,K\geq 1$ and that $\kappa\in(180/D,1/2)$. Let 
\begin{align}
    S&:= \{\mbox{Every eigenvector } v \mbox{ of } H \mbox{ satisfies } \|v_I\| \geq (\kappa s)^9 \|v\| \mbox{ for all } I \subset [D] \mbox{ with } |I|\geq \kappa D\},\\
    E&:=\left\{Q \in \mathbb{R}^{D\times D}: \mathbb{P}\left(\mathcal{B}_{\RE H+iQ,2J}\right)\geq \frac{1}{2}\right\}.
\end{align}
By the law of total expectation and the monotonicity of the conditional expectation we obtain that
\begin{subequations}
\begin{align}
    \mathbb{P}(S|\mathcal{B}_{H,2J}) &= \mathbb{E}\left(\mathbb{P}\left(S|\mathcal{B}_{H,2J}, \IM H\right)|\mathcal{B}_{H,2J}\right)\\
    &\geq \mathbb{E}\left(\mathbbm{1}_{\{\IM H \in E\}}\mathbb{P}\left(S|\mathcal{B}_{H,2J},\IM H\right)|\mathcal{B}_{H,2J}\right)\\
    &\geq \left(1-(Cs)^{\kappa D}\right)\mathbb{E}\left(\mathbbm{1}_{\{\IM H \in E\}}|\mathcal{B}_{H,2J}\right)\\
    &= \left(1-(Cs)^{\kappa D}\right)\mathbb{P}\left(\IM H \in E|\mathcal{B}_{H,2J}\right)\\
    &\geq \left(1-(Cs)^{\kappa D}\right)\mathbb{P}\left(\{\IM H \in E\}\cap \mathcal{B}_{H,2J}\right),
\end{align}
\end{subequations}
where $C=C(K,2J)\geq 1$ is the constant in Theorem \ref{thm: no gaps}, and Theorem \ref{thm: no gaps} itself was applied in the third line.
Next note that 
\begin{align}
    \mathcal{B}_{\RE H, J} \cap \mathcal{B}_{\IM H, J} \subset \mathcal{B}_{H,2J}\,,
\end{align}
and thus 
\begin{subequations}
\begin{align}
    \mathbb{P}\left(\{\IM H \in E\}\cap\mathcal{B}_{H,2J}\right) &\geq \mathbb{P}\left(\{\IM H \in E\}\cap\mathcal{B}_{\RE H,J}\cap\mathcal{B}_{\IM H,J}\right)\\
    &= \mathbb{P}\left(\mathcal{B}_{\RE H,J}\right) \mathbb{P}\left(\{\IM H \in E\}\cap\mathcal{B}_{\IM H,J}\right)\\
    &\geq (1-\eta)\, \mathbb{P}\left(\{\IM H \in E\}\cap\mathcal{B}_{\IM H,J}\right),
\end{align}
\end{subequations}
where we used that $\RE H$ and $\IM H$ are independent. Next observe that for $Q\in \mathbb{R}^{D\times D}$ such that $\|Q\|\leq J \sqrt{D}$, we have that 
\begin{align}
    \mathbb{P}\left(\mathcal{B}_{\RE H + iQ, 2J}\right) \geq \mathbb{P}\left(\mathcal{B}_{\RE H,J}\right) \geq 1-\eta \geq \frac{1}{2}.
\end{align}
It follows that $\mathcal{B}_{\IM H,J}\subset \{\IM H \in E\}$ and thus
\begin{align}
    \mathbb{P}\left(\{\IM H \in E\}\cap \mathcal{B}_{H,2J}\right) \geq (1-\eta)\, \mathbb{P}\left(\mathcal{B}_{\IM H,J}\right) \geq (1-\eta)^2.
\end{align}
Therefore, we finally obtain that
\begin{align}    \mathbb{P}\left(S|\mathcal{B}_{H,2J}\right) \geq \left(1-(Cs)^{\kappa D}\right) (1-\eta)^2\,,
\end{align}
which implies
\begin{subequations}
\begin{align}
    \mathbb{P}(S) &\geq \mathbb{P}\left(S\cap \mathcal{B}_{H,2J}\right) = \mathbb{P}\left(S|\mathcal{B}_{H,2J}\right)\mathbb{P}\left(\mathcal{B}_{H,2J}\right)\\
    &\geq \mathbb{P}\left(S|\mathcal{B}_{H,2J}\right) \mathbb{P}\left(\mathcal{B}_{\RE H,J}\right) \mathbb{P}\left(\mathcal{B}_{\IM H,J}\right)\\
    &\geq \left(1-(Cs)^{\kappa D}\right)(1-\eta)^4.
\end{align}
\end{subequations}
This shows that for complex random $D\times D$ matrices $H$ which satisfy the assumptions of this theorem with $J,K\geq 1$ (instead of $K,J>0$ such that $JK\geq 1$) and $\kappa\in(180/D,1/2)$ (instead of $\kappa \in(83/D,1/2)$), the following holds with probability at least $\left(1-(Cs)^{\kappa D}\right)(1-\eta)^4$: Every eigenvector $v$ of $H$ satisfies
\begin{align}
    \|v_I\| \geq (\kappa s)^9\|v\| \quad \mbox{ for all } I\subset [D] \mbox{ with } |I|\geq \kappa D,
\end{align}
where $C=C(K,2J)\geq 1$ and $C(K,2J)\sim \sqrt{K}J^{0.7}$.
\\
\\
\underline{2. Step ($J>0$)}:
As in the first step of the proof we assume that $\kappa \in (180/D,1/2)$. Note that if the density of a random variable $X$ is bounded by some constant $K$, then the density of $J^{-1}X$ is bounded by $JK$. Therefore,  applying the result of the first step to the matrix $\tilde{H}:= J^{-1}H$  immediately shows that for complex random $D\times D$ matrices $H$ which satisfy the assumptions of this theorem with $\kappa \in (180/D,1/2)$ (instead of $\kappa \in (83/D,1/2)$), the following holds with probability at least $\left(1-(\tilde{c}s)^{\kappa D}\right)(1-\eta)^4$: Every eigenvector $v$ of $H$ satisfies
\begin{align}
    \|v_I\| \geq (\kappa s)^9 \|v\| \quad \mbox{ for all } I\subset [D] \mbox{ with } |I|\geq \kappa D, 
\end{align}
where $\tilde{c} = C(KJ,2) \sim \sqrt{KJ}$.
\\
\\
\underline{3. Step (improved bound)}:
We first assume that the imaginary part of $H$ is fixed and that $K,J\geq 1$.
In the proof of the Invertibility Theorem 5.1 at the end of Section 5 in \cite{RV16}, they arrive at the upper bound
\begin{align}
    \left[ C_6 K J^{0.1}\kappa^{-0.05} \left(\frac{4Jt}{\tilde{c}_1\tau\kappa^{3/2}}\right)^{0.9}\right]^{\kappa D} + \left(C_4 \sqrt{\tilde{c}_2} K \tau\right)^{\kappa D}\label{eq: upper bound RV}
\end{align}
for the probability on the left-hand side of (5.18), where $t,\tau>0$ are arbitrary and where we used the notation for the constants introduced above. In their paper they choose $\tau=\sqrt{t}$ and then they use that they can assume that $t\leq 1$ to bound $t^{0.45}$ and $t^{0.5}$ by $t^{0.4}$. The result can be slightly improved by choosing $\tau=t^{9/19}$ instead. With this choice, the upper bound \eqref{eq: upper bound RV} becomes
\begin{align}
    \left[\left(\frac{4}{\tilde{c}_1}\right)^{0.9} C_6 KJ \kappa^{-1.4} t^{9/19}\right]^{\kappa D} + \left(C_4 \sqrt{\tilde{c}_2} K t^{9/19}\right)^{\kappa D} \leq \left(C_2 KJ \kappa^{-1.4} t^{9/19}\right)^{\kappa D},
\end{align}
where $C_2 = 2 \max\left\{\left(\frac{4}{\tilde{c}_1}\right)^{0.9}C_6, C_4\sqrt{\tilde{c}_2}\right\}$. This shows that the exponent of $t$ in Theorem~5.1 in \cite{RV16} can be changed from 0.4 to 9/19. 
This change has implications for the derivation of the no-gaps delocalization theorem from Theorem 5.1: After an application of Proposition 4.1 one now arrives at
\begin{align}
    \mathbb{P}\left(\mbox{Gap}(H,\kappa,t/8J)\; \mbox{and}\; \mathcal{B}_{H,J}\right) \leq 5\left(\frac{8J}{t}\right)^2 (e/\kappa)^{\kappa D} p_0
\end{align}
with $t\geq 0$ and $p_0=(2C_2 KJt^{9/19}\kappa^{-1.4})^{\kappa D/2}$.  Again we set $t=8J(\kappa s)^{\alpha}$ for some $\alpha>0$. Then a computation similar to the one in the proof of Theorem \ref{thm: no gaps} shows that
\begin{multline}
    \mathbb{P}\left(\mbox{Gap}(H,\kappa,(\kappa s)^{\alpha})\; \mbox{and}\; \mathcal{B}_{H,J}\right)\\
    \leq \left(\kappa^{-2\alpha/\kappa D+9\alpha/38-1.7} s^{-2\alpha/\kappa D+9\alpha/38-1} 5e \sqrt{2C_2 KJ}(8J)^{9/38}s\right)^{\kappa D}.
\end{multline} 
In order to obtain an upper bound for $\kappa^{-2\alpha/\kappa D+9\alpha/38-1.7} s^{-2\alpha/\kappa D+9\alpha/38-1}$ we have to determine $\alpha$ such that $-2\alpha/\kappa D+9\alpha/38-1.7 \geq 0$ since $\kappa <1/2$ and $s\leq 1$. Suppose again that $\kappa \geq \frac{\ell}{D}$ for some $\ell\in\mathbb{N}, \ell>8$, and $\alpha\geq \frac{323\ell}{5(9\ell-76)}$. Then
\begin{subequations}
\begin{align}
    -\frac{2\alpha}{\kappa D}+\frac{9\alpha}{38}-\frac{17}{10} &\geq -\frac{2\alpha}{\ell}+\frac{9\alpha}{38}-\frac{17}{10}\\
    &= \alpha\frac{9\ell-76}{38\ell}-\frac{17}{10}\\
    &\geq 0.
\end{align}
\end{subequations}
Thus, we arrive at
\begin{subequations}
\begin{align}
    \mathbb{P}\left(\mbox{Gap}(H,\kappa,(\kappa s)^{\alpha})\; \mbox{and}\; \mathcal{B}_{H,J}\right) &\leq \left(5e\sqrt{2C_2 KJ} (8J)^{9/38}s\right)^{\kappa D}\\
    &= \left(\frac{cs}{2}\right)^{\kappa D}
\end{align}
\end{subequations}
with $c=10 e\sqrt{2C_2KJ}(8J)^{9/38}$, i.e., $c \sim \sqrt{K} J^{14/19}$. Note that $\ell\mapsto \frac{323\ell}{5(9\ell-76)}$ for $\ell>8$ is monotonically decreasing and $\lim_{\ell\to\infty}\frac{323\ell}{5(9\ell-76)}=\frac{323}{45}$. Therefore we can choose $\alpha=8$ and $\ell=83$ to get $-2\alpha/\kappa D+9\alpha/38-1.7\geq 0$.
For matrices whose imaginary part is not fixed, we proceed as in step 1 to obtain the result for complex matrices. In order to relax the condition $J\geq 1$ to $J>0$ (and also $K\geq 1$ to $K>0$), one can again apply a scaling result as in step 2, where one replaces $K$ by $JK$ and $J$ by $1$, to obtain $c\sim \sqrt{KJ}\, 1^{14/19}=\sqrt{KJ}$. We thus arrive at the final result with $c_c=10 e\sqrt{2C_2}8^{9/38}$. Note that it depends on the value of $C_2$ whether $c_c\geq 1$ or merely $c_c>0$, however, we can assume without loss of generality that $c_c\geq 1$ since replacing $c_c$ by $\max\{10e\sqrt{2C_2}8^{9/38},1\}$ only results in a possibly smaller lower bound for the probability with which no-gaps delocalization at least holds.
\end{proof}

\begin{rmk}
\begin{enumerate}
    \item Theorem \ref{thm: no gaps extended} still holds if some entries of $\IM H$ are fixed. This observation will be important when we construct examples in Section \ref{subsec: ex} since for Hermitian matrices the diagonal entries of $\IM H$ are fixed to zero.
    \item In the proof of Theorem \ref{thm: no gaps extended} the choice $\alpha=8$ is the smallest one possible if one wants $\alpha\in\mathbb{N}$. Otherwise, any other value $\alpha> \frac{323}{45}\approx 7.18$ is possible, however, at the cost of a larger value for $\ell$ and therefore $\kappa$. 
\end{enumerate}
\end{rmk}

\begin{rmk}[A different proof strategy]
Another possibility (different from the scaling argument) to obtain no-gaps delocalization in the case that $0<J<1$ is to find out where the assumption $J\geq 1$ is used in the proof of Rudelson and Vershynin and to suitably modify the assumptions: The assumption $J\geq 1$ is used the first time (except for Section 5.1 to which we will turn at the end) in the proof of Lemma 5.7 (and Lemma 5.8 which follows from an application of Lemma 5.7). If $J\geq 1$, then $J/\sqrt{\kappa}>1$ and since the constant $C$ on the left-hand side of (5.15) and $K$ are also greater than 1, the upper bound on the probability is only non-trivial if $\theta^{1-2\mu}<1$. In view of how Lemma~5.7 is applied, we can safely assume that $\mu\leq 1/2$ and thus that $\theta<1$. In the proof of Lemma 5.7 as well as Lemma 5.8 it is used that $\theta\sqrt{\kappa}/J<1$ which is obviously fulfilled for this choice of parameter. If we allow that $J<1$ but require $\kappa<J^2$, then we can draw the same conclusions and argue in the same way to obtain Lemma 5.7 and Lemma 5.8. This means that the upper bound for $\kappa$ has to be replaced by $\min\{1/2,J^2\}$.

The next (and last) time the assumption $J\geq 1$ is used in the proof of the Invertibility Theorem 5.1 at the end of Section 5. Here the improved estimate that was presented in the third step in the proof of Theorem \ref{thm: no gaps extended} is necessary. 

Note that the assumption $JK\geq 1$ is not needed in this case and we have to assume again that $K\geq 1$. However, in our examples below, we always need that $JK\geq 1$. Moreover, note that we cannot assume anymore that the constant $c\sim \sqrt{K} J^{14/19}$ in the proof of Theorem~\ref{thm: no gaps extended} is larger than one. An adaption of the proof of Theorem~\ref{thm: main gen} shows that with probability at least 
\begin{align}
    \left(1-\min\left\{c_c\sqrt{K}J^{14/19},\frac{1}{2}\right\}^{\min\{d_\nu,d_\mu,2DJ^2\}/2}\right)(1-\eta)^4
\end{align}
one finds the lower bound 
\begin{align}
M_{\mu\nu} \geq \left(\min\left\{1,\frac{1}{2c_c\sqrt{K}J^{14/19}}\right\}\min\left\{\frac{d_\mu}{2D},\frac{d_\nu}{2D},J^2\right\}\right)^{16} \max\left\{1,\frac{d_\nu}{d_\mu}\right\}.
\end{align}
Similarly, with probability at least $\left(1-\min\{c_c\sqrt{K}J^{14/19},1/2\}^{\min\{d_\mu/2,DJ^2\}}\right)(1-\eta)^4$ one has the lower bound
 \begin{align}
     |M_{\mu B}| \geq \max\left\{b_{\min}^{\pm},\left(\min\left\{\frac{d_\mu}{2D},J^2\right\}\min\left\{1,\frac{1}{2c_c \sqrt{K}J^{14/19}}\right\}\right)^{16}\frac{\tr(B^{\pm})}{d_\mu}\right\}-\min\left\{b_{\max}^{\mp},\frac{\tr(B^{\mp})}{d_\mu}\right\}.
\end{align}

This lower bound is sometimes slightly better and sometimes slightly worse than the one we presented in Theorem \ref{thm: main new} but yields more complicated expressions e.g. in Corollary \ref{cor: real dim rel} and we therefore stick to the simpler lower bound.
\end{rmk}

\subsection{Examples\label{subsec: ex}}

 In this section we give some examples covered by Theorem \ref{thm: no gaps extended}. We will see that among these matrices are matrices with a band structure in a basis that diagonalizes the projections onto the macro spaces to which a ``small'' Gaussian perturbation is added. As discussed in the introduction, we are interested in such kind of matrices since, in contrast to Wigner matrices or matrices from the GOE/GUE, they can describe systems in which a non-trivial equilibration process passing through multiple macro-states occurs.

\begin{thm}[Gaussian matrices, variances bounded by constants]\label{thm: GM bdd}
Let $A=(a_{ij})$ be a $D\times D$ random matrix with independent complex Gaussian entries with mean zero, i.e., all random variables $\RE a_{ij}$, $\IM a_{ij}$, $i,j\in[D]$, are independent and $\RE a_{ij}, \IM a_{ij} \sim \mathcal{N}(0,\sigma_{ij}^2/2)$ for some $\sigma_{ij}>0$. Let $\sigma_{ij}= \sigma_{ji}$, $\sigma_- := \min_{i,j} \sigma_{ij}$ and $\sigma_+ := \max_{i,j} \sigma_{ij}$. Then the Hermitian matrix $H:=\frac{1}{\sqrt{2}}(A+A^*)$ is Gaussian, more precisely, its entries $h_{ij}$ satisfy $\RE h_{ij}, \IM h_{ij}\sim\mathcal{N}(0,\sigma_{ij}^2/2)$ for $i\neq j$ and $h_{ii}=\RE h_{ii}\sim \mathcal{N}(0,\sigma_{ii}^2)$, and fulfills the assumptions in Theorem \ref{thm: no gaps extended} with parameters $K = \frac{1}{\sqrt{2\pi}\sigma_-}$ and $J=\frac{4}{\eta}\hat{C}\sigma_+$ for arbitrary $\eta \in (0,\frac{1}{2})$, where $\hat{C}$ is a certain universal constant with the property \eqref{ineq: Latala}.
\end{thm}

\begin{proof}
The matrix $H$ is obviously Hermitian, and since $\RE a_{ij}$ and $\RE a_{ji}$ are independent and normally distributed with zero expectation and variance $\sigma_{ij}^2 = \sigma_{ji}^2$ for $i\neq j$, we have
\begin{align}
    \RE h_{ij} = \frac{1}{\sqrt{2}} \left(\RE a_{ij} + \RE \overline{a_{ji}}\right) = \frac{1}{\sqrt{2}}\left(\RE a_{ij} + \RE a_{ji}\right)\sim\mathcal{N}(0,\sigma_{ij}^2/2).
\end{align}
Similarly, one has $\IM h_{ij}\sim\mathcal{N}(0,\sigma_{ij}^2/2)$ and
\begin{align}
    \RE h_{ii} = \sqrt{2} \RE a_{ii} \sim \mathcal{N}(0,\sigma_{ii}^2).
\end{align}
By construction, $\RE h_{ij}$ is independent of the rest of the entries of $\RE H$ except $\RE h_{ji}$, the same holds true for the imaginary parts and obviously $\RE H$ and $\IM H$ are independent.

Latala \cite[Theorem 2]{Latala05} showed that
\begin{align}
    \mathbb{E}\|(x_{ij})\| \leq \hat{C} \left(\max_i \sqrt{\sum_j \mathbb{E}x_{ij}^2} + \max_j \sqrt{\sum_i \mathbb{E}x_{ij}^2} + \sqrt[4]{\sum_{i,j}\mathbb{E}x_{ij}^4}\right)\label{ineq: Latala}
\end{align}
for any finite matrix of independent mean zero random variables $x_{ij}$, where $\hat{C}>0$ is a universal constant. Without loss of generality we can assume that $\hat{C}\geq1$.

With the help of \eqref{ineq: Latala} the expectation of the norm of $\RE H$ can be bounded in the following way:
\begin{subequations}
\begin{align}
    \mathbb{E}\|\RE H\| &\leq \frac{1}{\sqrt{2}}\mathbb{E}\left(\|\RE A\| + \|\RE A^*\|\right)\\
    &\leq \hat{C}\left(\max_i \sqrt{\sum_j \sigma_{ij}^2} + \max_j \sqrt{\sum_i \sigma_{ij}^2} + \sqrt[4]{\sum_{i,j} 3\sigma_{ij}^4}\right)\\
    &\leq  \hat{C}\left(2\sqrt{D\sigma_+^2} + \sqrt[4]{3D^2\sigma_+^4}\right)\\
    &\leq 4\hat{C}\sigma_+ \sqrt{D}.
\end{align}
\end{subequations}
For $0<\eta\leq 1/2$ as in Theorem \ref{thm: no gaps extended} we set $J:=\frac{4}{\eta}\hat{C}\sigma_+$ and find with the help of Markov's inequality that
\begin{align}
    \mathbb{P}\left(\mathcal{B}_{\RE H,J}\right) = 1-\mathbb{P}\left(\|\RE H\| >J\sqrt{D} \right)
    \geq 1-\frac{\mathbb{E}\|\RE H\|}{J\sqrt{D}} &\geq 1-\eta,
\end{align}
i.e., the boundedness event $\mathcal{B}_{\RE H,J}$ holds with probability at least $1-\eta$. In the same way we find that also $\mathcal{B}_{\IM H,J}$ holds with probability at least $1-\eta$.

Clearly, the densities of $\RE H$ are bounded by $K := \frac{1}{\sqrt{2\pi}\sigma_-}$. This is due to the fact that the density function $f$ of the normal distribution with mean $\mu$ and variance $\sigma^2$ attains its maximum at $x=\mu$ with $f(\mu) = \frac{1}{\sqrt{2\pi}\sigma}$. and $\frac{1}{\sqrt{2\pi}\sigma_{ij}}\leq \frac{1}{\sqrt{2\pi}\sigma_-}$ for all $i,j\in[D]$. Moreover, the condition $JK\geq 1$ is obviously fulfilled for our choice of the parameters.
\end{proof}

\begin{rmk}
By choosing the variance large close to the diagonal of the matrix and small far away from it, the matrix $H$ in Theorem \ref{thm: GM bdd} has some kind of band structure. For example, we can fix a monotonically decreasing function $g: [0,1] \to [\sigma_-,\sigma_+]$, a ``variance profile,'' and define
\begin{align}
    \sigma_{ij} := g\left(\frac{|i-j|}{D}\right)\qquad\forall i,j\in[D].
\end{align}
\end{rmk}

So far we only considered matrices whose (Gaussian) entries have mean zero. In the following we want to relax this condition and also allow entries with non-zero mean.

\begin{thm}[Gaussian matrices with non-zero mean]\label{thm: GM nonzero mean}
Let $H_0=(h^0_{ij})$ be a (deterministic) Hermitian matrix with $C_{H_0}$ as in \eqref{const: CH0} and let $V$ be the $D\times D$ random matrix defined in Theorem \ref{thm: GM bdd} (there~$H$). Define $H := H_0+V$. Then $H$ is Hermitian with (non-centered) Gaussian independent (up to conjugate symmetry) entries; more precisely, its entries $h_{ij}$ satisfy $\RE h_{ij} \sim \mathcal{N}(\RE h^0_{ij},\sigma_{ij}^2/2)$, $\IM h_{ij} \sim \mathcal{N}(\IM h^0_{ij}, \sigma_{ij}^2/2)$ for $i\neq j$ and $h_{ii} = \RE h_{ii} \sim \mathcal{N}(h^0_{ii},\sigma_{ii}^2)$. Moreover, $H$ satisfies the assumptions in Theorem~\ref{thm: no gaps extended} with $K=\frac{1}{\sqrt{2\pi}\sigma_-}$ and $J = \frac{1}{\eta}(4\hat{C}\sigma_++C_{H_0})$ with $\hat{C}$ as before.
\end{thm}

\begin{proof}
The claims concerning the distribution of the entries of $H$ are obvious. With the help of the computation in the proof of Theorem \ref{thm: GM bdd} we find that
\begin{subequations}
\begin{align}
    \mathbb{E}\|\RE H\| &\leq \mathbb{E}\|\RE H_0\| + \|\RE V\|\\
    &\leq \left(4\hat{C}\sigma_++C_{H_0}\right)\sqrt{D}
\end{align}
\end{subequations}
For $0<\eta\leq 1/2$ as in Theorem~\ref{thm: no gaps extended} we set $J:=\frac{1}{\eta}\left(4\hat{C}\sigma_++C_{H_0}\right)$ and obtain
\begin{align}
\mathbb{P}\left(\mathcal{B}_{\RE H,J}\right) \geq 1-\eta,
\end{align}
i.e., the boundedness event $\mathcal{B}_{\RE H,J}$ holds with probability at least $1-\eta$. The parameter $K$ has already been computed in Theorem~\ref{thm: GM bdd} and with this choice the condition $KJ\geq 1$ is again automatically fulfilled. This finishes the proof.
\end{proof}

Theorem \ref{thm: GM nonzero mean} covers, for example, the case that $H_0$ is a matrix with a band structure in a basis that diagonalizes the projections $P_\mu$ and $V$ is a Gaussian matrix with mean zero and small variances, i.e., a small Gaussian perturbation with entries small in the $L^2$-sense.

\section{More General Results and Proofs \label{sec: more gen and proofs}}

In this section we state and prove a more general version of Theorem \ref{thm: main new} (Theorem~\ref{thm: main gen} in Section~\ref{sec:5.1}), some corollaries thereof and prove Theorem \ref{thm: main new}. We then state a more general version of Theorem~\ref{thm: dyntyp rel 2} (Theorem~\ref{thm:11} in Section~\ref{sec:5.2}) and prove Theorem~\ref{thm: dyntyp rel 2}.

\subsection{Generalized Normal Typicality}\label{sec:5.1}

Let $H = (h_{ij})\in M_{D\times D}(\mathbb{C})$ be a random Hermitian matrix. Instead of Assumption~\ref{asm: Hamiltonian}, we now make the following weaker assumption (which follows if Assumption~\ref{asm: Hamiltonian} holds):

\begin{asm}\label{asm: ind. and cont. distr.}
The random variables $(\RE h_{ij})_{i\leq j}$ and $(\IM h_{ij})_{i<j}$ are mutually independent and continuously distributed. The densities $\varrho_{ij}^{\RE}$ of $\RE h_{ij}$ are bounded.
\end{asm}

Let $B^+$ and $B^-$ denote the positive and negative part of $B$ such that $B=B^+-B^-$. Recall that $d_\mu$ and $d_\nu$ denote the dimensions of the macro spaces $\Hilbert_\mu$ and $\Hilbert_\nu$.

\begin{thm} \label{thm: main gen}
Let $\varepsilon'\in(0,\frac{1}{2})$ and let $\mu$ be an arbitrary macro state. Let $B\in M_{D\times D}(\mathbb{C})$ be a Hermitian $D\times D$ matrix and let $H=(h_{ij})$ be a random Hermitian $D\times D$ matrix such that Assumption \ref{asm: ind. and cont. distr.} is satisfied. Let $K$ be the least upper bound for the densities $\varrho_{ij}^{\RE}$, i.e.
\begin{align}
    K:=\sup\bigcup_{i\leq j}\left\{\varrho_{ij}^{\RE}(x): x\in\mathbb{R}\right\}<\infty.
\end{align}
Let $d_\mu>\max\left\{166, 2|\log_2\varepsilon'|\right\}$. Moreover, let $\eta\in(0,\frac{1}{2})$ be the unique number that solves
\begin{align}
    1-\varepsilon' = \left(1-2^{-d_\mu/2}\right)(1-\eta)^4
\end{align}
and let $J^{\RE}, J^{\IM}\in(0,\infty)$ be the unique numbers such that
\begin{align}
    \mathbb{P}\left(\|\RE H\|\leq J^{\RE}\sqrt{D}\right) &= 1-\eta,\\
    \mathbb{P}\left(\|\IM H\|\leq J^{\IM}\sqrt{D}\right) &= 1-\eta.
\end{align}
Set $J:=\max\left\{K^{-1},J^{\RE},J^{\IM}\right\}$. Then with probability at least $1-\varepsilon'$,
\begin{align}
    |M_{\mu B}| \geq \max\left\{b_{\min}^{+},\left(\frac{d_\mu}{4c_c \sqrt{KJ}D}\right)^{16}\frac{\tr(B^{+})}{d_\mu}\right\}-\min\left\{b_{\max}^{-}, \frac{\tr(B^{-})}{d_\mu}\right\},\label{ineq: MmuB lb}
\end{align}
where $M_{\mu B}$ was defined in \eqref{eq: MmuB}, $b_{\min}^{+}$ and $b_{\max}^{-}$ denote the smallest and largest eigenvalue of $B^{+}$ and $B^-$ respectively and $c_c\geq 1$ is the constant in Theorem \ref{thm: no gaps extended}.
In particular, if $B=P_\nu$ for some macro state $\nu$, then
\begin{align}
    M_{\mu\nu} \geq \frac{d_\nu}{d_\mu}\left(\frac{d_\mu}{4c_c \sqrt{KJ}D}\right)^{16}.\label{ineq: lb1}
\end{align}
Moreover, if $d_\nu\geq 166$, then for any $\eta\in(0,\frac{1}{2})$ (and $J$ chosen as above) it holds with probability at least $\left(1-2^{-d_\nu/2}\right)(1-\eta)^4$ that
\begin{align}
    M_{\mu\nu} \geq \left(\frac{d_\nu}{4c_c\sqrt{KJ}D}\right)^{16}.\label{ineq: lb2}
\end{align}
\end{thm}

The lower bound for $|M_{\mu B}|$ is obviously nontrivial for positive (and negative) operators $B$ but also for operators with positive and negative eigenvalues provided that the spectrum satisfies certain assumptions, e.g. if $b_{\min}^{+} > b_{\max}^{-}$.

Note the second lower bound \eqref{ineq: lb2} is sharper than \eqref{ineq: lb1} if $d_\nu \geq d_\mu$.
By combining the lower bounds in Theorem \ref{thm: main gen} with the upper bound from Theorem \ref{thm: abs errors A}, keeping in mind that, with probability 1, $D_E=1$ and $d_E=D$, we immediately obtain the following corollary:

\begin{cor}\label{cor: rel err}
Let $\varepsilon, \delta,\kappa, T>0, \varepsilon' \in (0,\frac{1}{2})$ and let $\mu$ be an arbitrary macro state. Let $B\in M_{D\times D}(\mathbb{C})$ be a Hermitian $D\times D$ matrix and let $H$ be a random Hermitian $D\times D$ matrix such that Assumption \ref{asm: ind. and cont. distr.} is satisfied. Let $d_\mu > \max\{166,2|\log_2\varepsilon'|\}$ and let $K,J>0$ and $\eta\in(0,\frac{1}{2})$ be defined as in Theorem~\ref{thm: main gen}. Then with probability at least $1-\varepsilon'$, $(1-\varepsilon)$-most $\psi_0\in\mathbb{S}(\Hilbert_\mu)$ are such that for $(1-\delta)$-most $t\in [0,T]$ 
\begin{align}
    \frac{\bigl|\langle\psi_t|B|\psi_t\rangle - M_{\mu B}\bigr|}{|M_{\mu B}|} \leq \frac{4 \left(\frac{ G(\kappa)\|B\|}{\delta \varepsilon d_\mu}\left(1+\frac{8\log_2 D}{\kappa T}\right)\min\left\{\|B\|, \frac{\tr(|B|)}{d_\mu}\right\}\right)^{1/2}}{\max\left\{b_{\min}^{+},\left(\frac{d_\mu}{4c_c\sqrt{KJ}D}\right)^{16}\frac{\tr(B^{+})}{d_\mu}\right\}-\min\left\{b_{\max}^{-}, \frac{\tr(B^{-})}{d_\mu}\right\}},
\end{align} 
whenever the denominator of the right-hand side is positive; here $c_c\geq 1$ is the constant in Theorem \ref{thm: no gaps extended}.
\end{cor}

The next corollary for the special case that $B=P_\nu$ for some macro state $\nu$ follows from combining the lower bounds \eqref{ineq: lb1} and \eqref{ineq: lb2} for $M_{\mu\nu}$, which yield with probability at least $(1-2^{-d_\nu/2}-2^{-d_\mu/2})(1-\eta)^4$
\begin{align}
    M_{\mu\nu} \geq \frac{d_\nu}{\max\{d_\nu,d_\mu\}} \left(\frac{\max\left\{d_\nu,d_\mu\right\}}{4c_c\sqrt{KJ}D}\right)^{16} = \left(\frac{\max\left\{d_\nu,d_\mu\right\}}{4c_c\sqrt{KJ}D}\right)^{16} \min\left\{1,\frac{d_\nu}{d_\mu}\right\}, \label{ineq: Mmunu combined}
\end{align}
with the upper bound for the absolute errors from Theorem \ref{thm: GNT}, keeping again in mind that for random matrices with continuously distributed entries, it holds with probability 1 that $D_E=D_G=1$.

\begin{cor}[Generalized normal typicality: relative errors]\label{cor: GNT rel err}
Let $\varepsilon,\delta>0$, $\varepsilon' \in (0,\frac{1}{2})$ and let $\mu, \nu$ be two macro states such that $d_\mu,d_\nu>\max\left\{166,2|\log_2(\varepsilon'/\sqrt{2})|\right\}$. Let $H$ be a random Hermitian $D\times D$ matrix such that Assumption \ref{asm: ind. and cont. distr.} is satisfied and let $K>0$ be defined as in Theorem \ref{thm: main gen}. Moreover, let $\eta \in (0,\frac{1}{2})$ be the unique number that solves
\begin{align}
    1-\varepsilon' = \left(1-2^{-d_\mu/2}-2^{-d_\nu/2}\right)(1-\eta)^4
\end{align}
and let $J>0$ be defined as in Theorem \ref{thm: main gen}. Then with probability at least $1-\varepsilon'$, $(1-\varepsilon)$-most $\psi_0\in\mathbb{S}(\Hilbert_\mu)$ are such that for $(1-\delta)$-most $t\in[0,\infty)$
\begin{align}
    \frac{\bigl|\|P_\nu\psi_t\|^2-M_{\mu\nu}\bigr|}{M_{\mu\nu}} \leq \frac{4}{\sqrt{\varepsilon\delta \min\{d_\mu,d_\nu\}}} \left(\frac{4c_c\sqrt{KJ}D}{\max\{d_\mu,d_\nu\}}\right)^{16},\label{ineq: GNT rel}
\end{align}
where $c_c\geq 1$ is the constant in Theorem \ref{thm: no gaps extended}.
\end{cor}

We now give the proofs of Theorem \ref{thm: main gen} and of one of our main theorems, Theorem~\ref{thm: main new}.

\begin{proof}[Proof of Theorem \ref{thm: main gen}]
We can write the Hamiltonian as $H=\sum_n E_n|\phi_n\rangle\langle\phi_n|$, where $(\phi_n)$ is an orthonormal basis of eigenvectors of $H$ with eigenvalues $E_n\in\mathbb{R}$. Moreover, we obtain with the reverse triangle inequality that
\begin{align}
    |M_{\mu B}| &\geq \frac{1}{d_\mu}\left(\sum_n \langle\phi_n|P_\mu|\phi_n\rangle\langle\phi_n|B^{+}|\phi_n\rangle - \sum_m \langle\phi_m|P_\mu|\phi_m\rangle\langle\phi_m|B^{-}|\phi_m\rangle\right).\label{ineq: MmuA splitting}
\end{align}
Note that we used here the fact that, with probability 1, the eigenvalues of $H$ are distinct.
Set $\kappa := \frac{d_\mu}{2D}$ and $s:=\frac{1}{2c_c\sqrt{KJ}}$, where $c_c\geq 1$ is the constant in Theorem \ref{thm: no gaps extended}. With these choices all assumptions in Theorem \ref{thm: no gaps extended} are fulfilled\footnote{Note that the factor $1/2$ in the definition of $\kappa$ ensures that $\kappa <1/2$ since $d_\mu/D$ can be arbitrarily close (but not equal to) 1. This is because the trivial case in which there is only one macro space is excluded here since in this case we obviously have normal typicality and there is nothing to show at all.} and with
\begin{align}
    I_\mu = \{d_1+\dots+d_{\mu-1}+1,\dots,d_1+\dots+d_\mu\}
\end{align}
it follows that with probability at least
\begin{align}
\left(1-\left(c_c\sqrt{KJ}s\right)^{\kappa D}\right)(1-\eta)^4 = \left(1-2^{-d_\mu/2}\right)(1-\eta)^4 = 1-\varepsilon'    
\end{align}
it holds that
\begin{align}
    \langle\phi_n|P_\mu|\phi_n\rangle = \sum_{j\in I_\mu} |\phi_n(j)|^2 = \|(\phi_n)_{I_\mu}\|^2 \geq (\kappa s)^{16}.
\end{align}
Remember that we assume that the Hamiltonian is written in a basis that diagonalizes the projections onto the macro spaces as in Figure \ref{fig:bandmatrix}.
We find the following lower bounds for the first and upper bounds for the second sum in \eqref{ineq: MmuA splitting}:
\begin{subequations}
\begin{align}
    \sum_n \langle\phi_n|P_\mu|\phi_n\rangle\langle\phi_n|B^{+}|\phi_n\rangle &\geq b_{\textup{min}}^{+} d_\mu,\label{ineq: Tmain 1}\\
    \sum_n \langle\phi_n|P_\mu|\phi_n\rangle \langle\phi_n|B^{+}|\phi_n\rangle &\geq (\kappa s)^{16} \tr(B^{+}),\\
    \sum_m \langle\phi_m|P_\mu|\phi_m\rangle\langle\phi_m|B^{-}|\phi_m\rangle &\leq b_{\textup{max}}^{-} d_\mu,\label{ineq: Tmain 3}\\
    \sum_m \langle\phi_m|P_\mu|\phi_m\rangle\langle\phi_m|B^{-}|\phi_m\rangle &\leq \tr(B^{-})\label{ineq: Tmain 4}.
\end{align}
\end{subequations}
A combination of these bounds yields
\begin{align}
    |M_{\mu B}| &\geq \max\left\{b_{\min}^{+},(\kappa s)^{16}\frac{\tr(B^{+})}{d_\mu}\right\} - \min\left\{b_{\max}^{-},\frac{\tr(B^{-})}{d_\mu}\right\}\\
    &= \max\left\{b_{\min}^{+}, \left(\frac{d_\mu}{4c_c\sqrt{KJ}D}\right)^{16} \frac{\tr(B^{+})}{d_\mu}\right\} - \min\left\{b_{\textup{max}}^{-},\frac{\tr(B^{-})}{d_\mu}\right\}.
\end{align}
In particular, if $B=P_\nu$ for some macro state $\nu$, then $B^-=0$, $b_{\min}^+=0$, $\tr(B^+)=d_\nu$ and thus
\begin{align}
    M_{\mu\nu} \geq \frac{d_\nu}{d_\mu} \left(\frac{d_\mu}{4c_c\sqrt{KJ}D}\right)^{16}.
\end{align}
Note that here no absolute value on the left-hand side is needed since it immediately follows from the definition of the $M_{\mu\nu}$ that $M_{\mu\nu}\geq 0$.
If $d_\nu\geq 166$, we set $\kappa := \frac{d_\nu}{2D}$ and obtain with Theorem \ref{thm: no gaps extended} with probability at least $\left(1-2^{-d_\nu/2}\right)(1-\eta)^4$
\begin{align}
    \langle\phi_n|P_\nu|\phi_n\rangle \geq (\kappa s)^{16} = \left(\frac{d_\nu}{4c_c\sqrt{KJ}D}\right)^{16}
\end{align}
and therefore
\begin{align}
    M_{\mu\nu} = \frac{1}{d_\mu}\sum_n \langle\phi_n|P_\nu|\phi_n\rangle \langle\phi_n|P_\mu|\phi_n\rangle \geq \left(\frac{d_\nu}{4c_c\sqrt{KJ}D}\right)^{16}.
\end{align}
\end{proof}

\begin{proof}[Proof of Theorem \ref{thm: main new}]
First note that since $H$ satisfies Assumption \ref{asm: Hamiltonian} it obviously also satisfies Assumption \ref{asm: ind. and cont. distr.} and that $d_\nu,d_\mu>\max\{166,2|\log_2(\varepsilon'/\sqrt{2})|\}$ by assumption. Let $\eta\in(0,\frac{1}{2})$ be the unique number that solves
\begin{align}
    1-\varepsilon'= \left(1-2^{-2_\mu/2}-2^{-d_\nu/2}\right)(1-\eta)^4.\label{eq: 1-eps'}
\end{align}
Then it follows from Theorem \ref{thm: no gaps extended} as in the proof of Theorem \ref{thm: main gen} that with probability at least $1-\varepsilon'$,
\begin{align}
    M_{\mu\nu} \geq \left(\frac{\max\{d_\nu,d_\mu\}}{4c_c\sqrt{KJ}D}\right)^{16} \min\left\{1,\frac{d_\nu}{d_\mu}\right\},\label{ineq: Mmunu combined 2}
\end{align}
see also \eqref{ineq: Mmunu combined}. In order to arrive at the form \eqref{mainnewMmunu}, note first that with $C_{H_0}$ as defined in \eqref{const: CH0} and with the help of Theorem~\ref{thm: GM nonzero mean}, the factor $KJ$ in \eqref{ineq: Mmunu combined 2} is given by 
\begin{align}
    KJ = \frac{1}{\sqrt{2\pi}\sigma_- \eta}\left(4\hat{C}\sigma_++C_{H_0}\right).\label{ineq: bound constant}
\end{align}
Next we derive a lower bound for $\eta$ in terms of $\varepsilon'$ in order to eliminate $\eta$ from \eqref{ineq: bound constant}. Therefore observe that $d_\nu,d_\mu > 4|\log_2 (\varepsilon'/\sqrt{2})| = -4\log_2(\varepsilon'/\sqrt{2})$ and with \eqref{eq: 1-eps'} it follows that
\begin{align}
    1-\varepsilon' \geq \left(1-\varepsilon'^2\right)(1-\eta)^4.
\end{align}
Solving for $\eta$ yields
\begin{align}
    \eta \geq 1- \frac{1}{\left(1+\varepsilon'\right)^{1/4}} \geq \frac{\varepsilon'}{6}.
\end{align}
In the last step it was used that for $f,g:[0,\frac{1}{2}]\to\mathbb{R}$, $f(x)=1-\frac{1}{(1+x)^{1/4}}$, $g(x)=\frac{x}{6}$ one has $f\geq g$ which can easily be shown by standard arguments. Thus with \eqref{ineq: bound constant} we find
\begin{align}
    c_c^2 KJ&\leq \frac{6 c_c^2}{\sqrt{2\pi}\sigma_- \varepsilon'}\left(4\hat{C}\sigma_++C_{H_0}\right) =: \frac{1}{16\varepsilon'} \frac{c_+\sigma_++C_{H_0}}{c_-\sigma_-}
\end{align}
with $c_- := \frac{\sqrt{2\pi}}{96c_c^2}$ and $c_+ :=4\hat{C}$, i.e., $c\leq 1/(4\sqrt{\varepsilon' C_{\sigma}})$. Inserting this estimate into \eqref{ineq: Mmunu combined 2} finishes the proof.
\end{proof}

\subsection{Dynamical Typicality}
\label{sec:5.2}

In order to find an upper bound for the relative error in the dynamical typicality theorem, Theorem \ref{thm: dyntyp phys}, as well, we need a lower bound for $|w_{\mu B}(t)| = |\mathbb{E}_\mu\langle\psi_t|B|\psi_t\rangle|$. Without any further assumptions on the Hamiltonian we immediately find the following proposition:

\begin{prop}\label{thm: dyntyp rel 1}
Let $B\in\mathcal{L}(\Hilbert)$ be Hermitian such that $b:=\max\{b_{\min}^+-b_{\max}^-,b_{\min}^- - b_{\max}^+\}>0$, where $b_{\min}^\pm$ and $b_{\max}^\pm$ are again the smallest and largest eigenvalues of $B^\pm$ (the positive and negative parts of $B=B^+-B^-$). Let $\varepsilon>0$, $t\in[0,\infty)$, and let $\mu$ be an arbitrary macro state. Then
\begin{align}
\bigl|\mathbb{E}_\mu\langle\psi_t|B|\psi_t\rangle\bigr| \geq b,
\end{align}
and therefore $(1-\varepsilon)$-most $\psi_0\in\mathbb{S}(\Hilbert_\mu)$ are such that
\begin{align}
    \frac{\bigl|\langle\psi_t|B|\psi_t\rangle - \mathbb{E}_{\mu}\langle\psi_t|B|\psi_t\rangle\bigr|}{\bigl|\mathbb{E}_\mu\langle\psi_t|B|\psi_t\rangle\bigr|} \leq b^{-1} \cdot(\mbox{right-hand side of \eqref{ineq: dyntyp1}}).
\end{align}
Moreover, for every $\mu$ and $B$, every $T>0$, and $(1-\varepsilon)$-most $\psi_0\in\mathbb{S}(\Hilbert_\mu)$,
\begin{align}
    \frac{1}{T}\int_0^T \frac{\bigl|\langle\psi_t|B|\psi_t\rangle - \mathbb{E}_\mu\langle\psi_t|B|\psi_t\rangle\bigr|^2}{\bigl|\mathbb{E}_\mu\langle\psi_t|B|\psi_t\rangle\bigr|^2}\, dt \leq \frac{\|B\|^2}{b^2\varepsilon d_\mu}.\label{ineq: dyntyp curve simple}
\end{align}
\end{prop}

\begin{proof}
Let $(\varphi_k)$ be an orthonormal basis of $\Hilbert_\mu$ and define $\varphi_{k,t} := e^{-itH}\varphi_k$. Then
\begin{subequations}
\begin{align}
\bigl|\mathbb{E}_\mu\langle\psi_t|B|\psi_t\rangle\bigr| &= \frac{1}{d_\mu}\bigl|\tr(P_\mu e^{itH} B e^{-itH}) \bigr|\\
    &= \frac{1}{d_\mu}\biggl|\sum_k\langle\varphi_{k,t}|B|\varphi_{k,t}\rangle \biggr|\\
    &\geq \frac{1}{d_\mu}\max\left\{\sum_k \langle\varphi_{k,t}|B^{+}-B^- |\varphi_{k,t}\rangle,\sum_k \langle\varphi_{k,t}|B^{-}-B^+|\varphi_{k,t}\rangle \right\}\\
    &\geq \max\{b_{\min}^{+} - b_{\max}^{-},b_{\min}^- - b_{\max}^+\},
\end{align}
\end{subequations}
i.e., $\bigl|\mathbb{E}_\mu\langle\psi_t|B|\psi_t\rangle\bigr|\geq b$.
The remaining claims now follow immediately from Theorem~\ref{thm: dyntyp phys}.
\end{proof}

Proposition~\ref{thm: dyntyp rel 1} yields a good lower bound for $\bigl|\mathbb{E}_\mu\langle\psi_t|B|\psi_t\rangle\bigr|$ and therefore useful upper bounds for the relative errors if, for example, $B$ is a positive or negative operator (and in this case, $b=b_{\min}^+>0$ and $b=b_{\min}^->0$ respectively). If $B$ has both positive and non-positive eigenvalues, the bounds are only useful if the spectrum of $B$ is rather special since in most cases one has $b\leq 0$. In particular, if $B=P_\nu$ for some macro state $\nu$ we have $b=0$ and Proposition~\ref{thm: dyntyp rel 1} is not applicable. However, with the help of the corrected and improved no gaps delocalization, Theorem \ref{thm: no gaps extended}, we are able to prove an upper bound for the comparative errors which also applies to more general operators and in particular to the case $B=P_\nu$:

\begin{thm}\label{thm:11}
    Let $\varepsilon>0$, $\varepsilon'\in (0,\frac{1}{2})$ and let $\mu$ and $\nu$ be macro states such that $d_\mu,d_\nu>\max\{166,4|\log_2\varepsilon'|\}$. Let $B$ be a Hermitian $D\times D$ matrix and let $H$ be a random Hermitian $D\times D$ matrix as in Theorem~\ref{thm: main gen}. Then with probability at least $1-\varepsilon'$, for each $t\in\mathbb{R}$, $(1-\varepsilon)$-most $\psi_0\in\mathbb{S}(\Hilbert_\mu)$ are such that
    \begin{multline}
        \frac{\bigl|\langle\psi_t|B|\psi_t\rangle - \mathbb{E}_\mu\langle\psi_t|B|\psi_t\rangle\bigr|}{\left|\overline{\langle\psi_t|B|\psi_t\rangle }\right|} \\ 
        \leq \mathrm{LB(B,\psi)}^{-1} \min\left\{\frac{\sqrt{2}\|B\|}{\sqrt{\varepsilon d_\mu}},\sqrt{\frac{2\|B\|\tr(|B|)}{\varepsilon d_\mu^2}}, \sqrt{\frac{18\pi^3 \log(8/\varepsilon)}{d_\mu}\|B\|}\right\},
    \end{multline}
    whenever
    \begin{align}
        \mathrm{LB}(B,\psi) &:= \max\left\{b_{\min}^{+}, \left(\frac{d_\mu}{4c_c\sqrt{KJ}D}\right)^{16} \frac{\tr(B^{+})}{d_\mu}\right\} - \min\left\{b_{\textup{max}}^{-},\frac{\tr(B^{-})}{d_\mu}\right\}\nonumber\\
        &\quad- \sqrt{2}\left(\frac{\|B\|}{\varepsilon d_\mu}\min\left\{\|B\|,\frac{\tr(|B|)}{d_\mu}\right\}\right)^{1/2}
    \end{align}
is positive.  
Moreover, for every $\mu$ and $B$, every $T>0$, and $(1-\varepsilon)$-most $\psi_0\in\mathbb{S}(\Hilbert_\mu)$,
\begin{align}
    \frac{1}{T}\int_0^T \frac{\bigl|\langle\psi_t|B|\psi_t\rangle - \mathbb{E}_\mu\langle\psi_t|B|\psi_t\rangle \bigr|^2}{\left|\overline{\langle\psi_t|B|\psi_t\rangle}\right|^2}\,dt \leq \frac{2\|B\|^2}{\mathrm{LB}(B,\psi)^2\varepsilon d_\mu}.
\end{align}
\end{thm}

\begin{proof}
  It follows from (124) and (125) in \cite{TTV22-physik} that $\left(1-\frac{\varepsilon}{2}\right)$-most $\psi_0\in\mathbb{S}(\Hilbert_\mu)$ are such that
  \begin{align}
      \left|\overline{\langle\psi_t|B|\psi_t\rangle} - M_{\mu B}\right| \leq \sqrt{2}\left(\frac{\|B\|}{\varepsilon d_\mu}\min\left\{\|B\|,\frac{\tr(|B|)}{d_\mu}\right\}\right)^{1/2}.
  \end{align}
  Therefore it follows from the triangle inequality and Theorem \ref{thm: main gen} that with probability at least $1-\varepsilon'$, $\left(1-\frac{\varepsilon}{2}\right)$-most $\psi_0\in\mathbb{S}(\Hilbert_\mu)$ are such that
  \begin{align}
      \left|\overline{\langle \psi_t|B|\psi_t\rangle}\right| &\geq |M_{\mu B}| - \sqrt{2}\left(\frac{\|B\|}{\varepsilon d_\mu}\min\left\{\|B\|,\frac{\tr(|B|)}{d_\mu}\right\}\right)^{1/2}\\
      &\geq \max\left\{b_{\min}^{+}, \left(\frac{d_\mu}{4c_c\sqrt{KJ}D}\right)^{16} \frac{\tr(B^{+})}{d_\mu}\right\} - \min\left\{b_{\textup{max}}^{-},\frac{\tr(B^{-})}{d_\mu}\right\}\nonumber\\
      &\quad- \sqrt{2}\left(\frac{\|B\|}{\varepsilon d_\mu}\min\left\{\|B\|,\frac{\tr(|B|)}{d_\mu}\right\}\right)^{1/2}\\
      &= \mathrm{LB}(B,\psi).\label{ineq: LB(B,psi)}
  \end{align}
It follows from Theorem~\ref{thm: dyntyp phys} that $\left(1-\frac{\varepsilon}{2}\right)$-most $\psi_0\in\mathbb{S}(\Hilbert_\mu)$ are such that
\begin{align}
    \bigl|\langle\psi_t|B|\psi_t\rangle - \mathbb{E}_\mu\langle\psi_t|B|\psi_t\rangle \bigr| \leq \min\left\{\frac{\sqrt{2}\|B\|}{\sqrt{\varepsilon d_\mu}},\sqrt{\frac{2\|B\|\tr(|B|)}{\varepsilon d_\mu^2}}, \sqrt{\frac{18\pi^3 \log(8/\varepsilon)}{d_\mu}\|B\|}\right\}. \label{ineq: dyntyp phys eps/2}
\end{align}
A combination of \eqref{ineq: LB(B,psi)} and \eqref{ineq: dyntyp phys eps/2} yields the first claim. Similarly, a combination of \eqref{ineq: LB(B,psi)} and \eqref{ineq: dyntyp2} in Theorem~\ref{thm: dyntyp phys} gives the second claim. 
\end{proof}

Note that for large $d_\mu$ (and if $\|B\|$ does not depend on $d_\mu$) the lower bound for $\bigl|\overline{\langle\psi_t|B|\psi_t\rangle}\bigr|$, $\mathrm{LB}(B,\psi)$, is up to a small error equal to the lower bound for $|M_{\mu B}|$ that was proved in Theorem~\ref{thm: main gen}.

Finally, we give the proof of Theorem \ref{thm: dyntyp rel 2} which yields a somewhat better bound than the previous theorem in the case that $B$ is a projection $P_\nu$ onto some macro space $\Hilbert_\nu$:

\begin{proof}[Proof of Theorem \ref{thm: dyntyp rel 2}]
Let $(\phi_n)$ be an orthonormal basis of eigenvectors of $H$. With $\psi_0=\sum_n c_n\phi_n$ we find that
\begin{align} 
\overline{\langle\psi_t|P_\nu|\psi_t\rangle} = \sum_{n,m} c_n^* c_m \overline{e^{it(E_n-E_m)}} \langle\phi_n|P_\nu|\phi_m\rangle = \sum_n |c_n|^2 \langle\phi_n|P_\nu|\phi_n\rangle.
\end{align}
Since $\sum_n |c_n|^2=1$ and with the help of Theorem \ref{thm: no gaps extended} we find with probability at least $1-\varepsilon'$ that
\begin{align}
\overline{\langle\psi_t|P_\nu|\psi_t\rangle} \geq \left(\frac{d_\nu}{4c_c\sqrt{KJ}D}\right)^{16} \geq \left(\frac{\sqrt{\varepsilon' C_{\sigma}}d_\nu}{D}\right)^{16}
\end{align}
see also the proof of Theorem \ref{thm: main gen} for the first step. In the second step we used that $c_c\sqrt{KJ}\leq 1/(4\sqrt{\varepsilon' C_{\sigma}})$, which can be shown similarly as in the proof of Theorem \ref{thm: main new}.
This implies together with Theorem~\ref{thm: dyntyp phys} (by using the first two bounds in \eqref{ineq: dyntyp1}) that $(1-\varepsilon)$-most $\psi_0\in\mathbb{S}(\Hilbert_\mu)$ are such that
\begin{align}
    \frac{\bigl|\|P_\nu\psi_t\|^2-\mathbb{E}_\mu\|P_\nu\psi_t\|^2 \bigr|}{\overline{\|P_\nu\psi_t\|^2}} \leq \frac{1}{\sqrt{\varepsilon}} (C_{\sigma}\varepsilon')^{-8} \frac{1}{d_\mu} \sqrt{\min\left\{d_\mu,d_\nu\right\}} \left(\frac{D}{d_\nu}\right)^{16}.\label{ineq: lb TA Pnupsi}
\end{align}
Inserting the definition of $s_\mu, s_\nu$ and $s_{\mathrm{mc}}$ thus proves the first claim. For the second claim observe that it follows from Theorem~\ref{thm: dyntyp phys} that for every $T>0$, $(1-\varepsilon)$-most $\psi_0\in\mathbb{S}(\Hilbert_\mu)$ are such that 
\begin{align}
    \frac{1}{T} \int_0^T \bigl|\|P_\nu\psi_t\|^2-\mathbb{E}_\mu\|P_\nu\psi_t\|^2 \bigr|^2\, dt \leq \frac{1}{\varepsilon d_\mu}.
\end{align}
Together with \eqref{ineq: lb TA Pnupsi} this implies that with probability $1-\varepsilon'$, for every $T>0$, $(1-\varepsilon)$-most $\psi_0\in\mathbb{S}(\Hilbert_\mu)$ are such that
\begin{align}
    \frac{1}{T}\int_0^T \frac{\bigl| \|P_\nu\psi_t\|^2-\mathbb{E}_\mu\|P_\nu\psi_t\|^2\bigr|^2}{\overline{\|P_\nu\psi_t\|^2}}\, dt \leq \frac{1}{\varepsilon} (C_{\sigma}\varepsilon')^{-16} \frac{1}{d_\mu} \left(\frac{D}{d_\nu}\right)^{32}.
\end{align}
Now also the second claim follows immediately from the definition of $s_\mu,s_\nu$ and $s_{\mathrm{mc}}$.
\end{proof}

\section{An Improved Lower Bound for $\boldsymbol{M}_{\boldsymbol{\mu}\,\textup{eq}}$, $\boldsymbol{M}_{\textup{eq}\,\boldsymbol{\nu}}$ and $\boldsymbol{M}_{\textup{eq}\,\textup{eq}}$ for Wigner-Type Matrices \label{sec: Improved LB}}

We expect that in many cases a significantly stronger lower bound for the $M_{\mu\nu}$ than the one obtained with the help of our improved version of the no gaps delocalization result from Rudelson and Vershynin \cite{RV16}, Theorem~\ref{thm: main new}, should hold. 
More precisely, it is expected that for band matrices $H$ with sufficiently large band width, the eigenvectors are delocalized, and in that situation we expect 
that $M_{\mu\nu} \approx \frac{d_\nu}{D}$.

In this section, we show with the help of a result from Ajanki, Erdős and Krüger \cite{AEK17} concerning the delocalization of eigenvectors that for a special class of random matrices a much better lower bound for $M_{\mu\,\textup{eq}}$, $M_{\textup{eq}\,\nu}$ and $M_{\textup{eq}\,\textup{eq}}$ can be obtained provided that the equilibrium macro space is very dominant (in a sense made precise below). These matrices do not have a band structure. 
However, the results strengthen the expectation that often the lower bound in Theorem~\ref{thm: main new} can be significantly improved. We quote the relevant statement from \cite{AEK17} as Theorem~\ref{thm: AEK deloc} and obtain from it lower bounds for $M_{\mu\nu}$ in Theorem~\ref{thm: lb AEK}.

For using the result of \cite{AEK17}, we need a certain shift of perspective. So far in this paper, we always regarded $D$ as fixed and considered a single, randomly chosen $D\times D$ matrix $H$. In contrast, in this section and the next, we will consider a \emph{sequence} of random matrices $H^{(D)}$, one for every $D\in\mathbb{N}$. In fact, our reasoning also applies if we merely allow an infinite set of $D$ values (say, all powers of 2), and a random matrix $H^{(D)}$ for every $D$ from that set; but for simplicity, we will pretend we have an $H^{(D)}$ for every $D\in\mathbb{N}$. More precisely, we assume that, for every $D\in\mathbb{N}$, we are given a probability distribution $\PPP^{(D)}$ over the Hermitian $D\times D$ matrices. We will show that, under suitable assumptions on $\PPP^{(D)}$, certain estimates hold for \emph{sufficiently large $D$}, but we are not necessarily able to make explicit how large $D$ has to be.

In order to state the result of \cite{AEK17}, we need their notion of \textit{stochastic domination}:

\begin{defn}[Stochastic domination] \label{defn: stoch dom} 
For two sequences $X=(X^{(D)})_D$ and $Y=(Y^{(D)})_D$ of non-negative random variables we say that $X$ is stochastically dominated by $Y$ if there exists a function $D_0:(0,\infty)^2\to\mathbb{N}$ such that for all $\tau>0$ and $\alpha>0$,
\begin{align}
    \mathbb{P}\left(X^{(D)} > D^{\tau} Y^{(D)} \right) \leq D^{-\alpha} \qquad \forall D\geq D_0(\tau,\alpha).
\end{align}
In this case we write $X\prec Y$.
\end{defn}

That is, $X\prec Y$ means that for large $D$ it has high probability that $X$ is not much larger than $Y$. The key fact for our purposes will be that under certain assumptions on $\PPP^{(D)}$, every eigenvector $\phi_n$ of $H=H^{(D)}\sim \PPP^{(D)}$ satisfies
\be\label{AEK3}
\|\phi_n\|_\infty \prec \frac{1}{\sqrt{D}}
\ee
(more precise statement around \eqref{AEK2} below). That is, each component of $\phi_n$ is not much larger than $1/\sqrt{D}$ (the magnitude that a component would have to have if all components had the same magnitude). Since if some components were much smaller than $1/\sqrt{D}$, others would have to be larger, this also entails that not a large fraction of the components can be much smaller than $1/\sqrt{D}$. But this still allows that a \emph{small} fraction of the components could be arbitrarily small; for example, if $d_\nu$ (despite being a large number) is a small fraction of $D$, then all components of $\phi_n$ in $\Hilbert_\nu$ could be arbitrarily small, so $\langle \phi_n|P_\nu|\phi_n\rangle$ could be arbitrarily small, and then the expression \eqref{Mmunuexpression} suggests that $M_{\mu\nu}$ could be arbitrarily small, and we would not obtain a useful lower bound for $M_{\mu\nu}$. In fact, \eqref{AEK3} provides such a bound only if either $d_\mu$ or $d_\nu$ is sufficiently close to $D$, as the detailed analysis below confirms.

We now turn to describing the matrices considered in \cite{AEK17} and in this section.

\begin{asm}\label{asm: S1 (flat)}
For every $D$, $H^{(D)}=(h_{ij})\sim \PPP^{(D)}$ is a Hermitian $D\times D$ matrix of \textit{Wigner type}, i.e., its entries $h_{ij}$ are centered random variables and the entries $(h_{ij})_{i\leq j}$ are independent. The matrix of variances $S=(\sigma^2_{ij})$, defined by
\begin{align}
    \sigma^2_{ij} := \mathbb{E}|h_{ij}|^2 \,,
\end{align}
is \textit{flat}, i.e.,
 \begin{align}
     \sigma^2_{ij} \leq \frac{1}{D}, \qquad i,j=1,\dots,D.
 \end{align}
\end{asm}

Let $p,P>0$ and $L\in\mathbb{N}$ be parameters independent of $D$ and let $\mu=(\mu_1,\mu_2,\dots)$ be a sequence of non-negative real numbers. 

\begin{asm}\label{asm: S2 (un prim)}
 The matrix $S$ is \textit{uniformly primitive}, i.e.
 \begin{align}
     (S^L)_{ij} \geq \frac{p}{D}, \qquad i,j=1,\dots,D.
 \end{align}
\end{asm}

It can be shown that the corresponding 
\emph{vector Dyson equation}
\begin{align}
    -\frac{1}{m_i(z)} = z + \sum_{j=1}^D \sigma^2_{ij} m_j(z), \qquad
     \mbox{for all}\; i=1,\dots,D \; \mbox{and}\; z\in\mathbb{H} \label{eq: QVE}
\end{align}
for a function $m=(m_1,\dots,m_D):\mathbb{H}\to\mathbb{H}^D$ on the complex upper half plane $\mathbb{H} = \{z\in\mathbb{C}: \IM z>0\}$ has a unique solution \cite{AEK19}.

\begin{asm}\label{asm: S3 (bdd sol of QVE)} 
The matrix $S$ induces a bounded solution of the vector Dyson equation, i.e., the unique solution of \eqref{eq: QVE} corresponding to $S$ is bounded, 
\begin{align}
    |m_i(z)| \leq P, \qquad i=1,\dots,D, \quad z\in\mathbb{H}.
\end{align}
\end{asm}

\begin{asm}\label{asm: S4 bdd mom}
The entries $h_{ij}$ of the random matrix $H$ have \textit{bounded moments}, i.e.
\begin{align}
    \mathbb{E}|h_{ij}|^k \leq \mu_k \sigma_{ij}^{k}, \qquad k\in\mathbb{N}, \quad i,j = 1,\dots,D.
\end{align}
\end{asm}

Sufficient conditions for Assumption~\ref{asm: S3 (bdd sol of QVE)} are given in \cite{AEK19} Theorem 6.1. The reason we make these assumptions is that under these conditions, 
Ajanki, Erdős and Krüger (2017) \cite{AEK17} proved the following theorem concerning the delocalization of the eigenvectors of $H$: 

\begin{thm}[\cite{AEK17} Corollary~1.14]\label{thm: AEK deloc}
Consider a sequence $(\PPP^{(D)})_{D\in\mathbb{N}}$ of probability distributions over the Hermitian $D\times D$ matrices with $H^{(D)}\sim \PPP^{(D)}$ satisfying Assumptions \ref{asm: S1 (flat)}-\ref{asm: S4 bdd mom}. 
Let $E_1^{(D)}\leq \ldots \leq E_D^{(D)}$ be the eigenvalues of the random matrix $H^{(D)}$ and $\phi_n^{(D)} \in \mathbb{C}^D$ a normalized eigenvector of $H^{(D)}$ with eigenvalue $E_n^{(D)}$. 
Then for every sequence of unit vectors $b^{(D)}\in\mathbb{C}^D$ and every sequence $n^{(D)}\in\{1,\ldots,D\}$,
\begin{align}
   \left|\bigl\langle b^{(D)}\big| \phi_{n^{(D)}}^{(D)} \bigr\rangle\right| \prec \frac{1}{\sqrt{D}} \,. 
\end{align}
In particular, the eigenvectors are completely delocalized, i.e.,
\be\label{AEK2}
\bigl\|\phi_{n^{(D)}}^{(D)}\bigr\|_\infty \prec \frac{1}{\sqrt{D}} \,.
\ee
The function $D_0(\tau,\alpha)$ implicit in the $\prec$ symbol in \eqref{AEK2} depends only on the constants $p,P,L,(\mu_k)_{k\in\mathbb{N}}$ from Assumptions \ref{asm: S1 (flat)}-\ref{asm: S4 bdd mom}.
\end{thm}

With the help of Theorem \ref{thm: AEK deloc} we find the following lower bounds for the $M_{\mu B}$:

\begin{thm}[Lower bounds for $|M_{\mu B}|$]\label{thm: lb AEK} 
Consider a sequence $(\PPP^{(D)})_{D\in\mathbb{N}}$ of probability distributions over the Hermitian $D\times D$ matrices with $H^{(D)}\sim \PPP^{(D)}$ satisfying Assumptions \ref{asm: S1 (flat)}-\ref{asm: S4 bdd mom} and let $D_0:(0,\infty)^2\to\mathbb{N}$ be the function provided by Theorem~\ref{thm: AEK deloc} for \eqref{AEK2}. Let $\tau>0,\alpha>1$,  
$D\geq D_0(\tau,\alpha)$, and let $B$ be a Hermitian $D\times D$ matrix. Then 
with probability at least $1-D^{-\alpha+1}$ it holds for every macro state $\mu$ that
\begin{align}
    |M_{\mu B}| \geq \max\left\{b_{\min}^{+}, \frac{\tr(B^{+})}{d_\mu}\left(1-\frac{D-d_\mu}{D^{1-2\tau}}\right)\right\} - \min\left\{b_{\max}^{-}, \frac{\tr(B^{-})}{d_\mu}\right\}.
\end{align}
In particular, if $B=P_\nu$ for some macro state $\nu$, then
\begin{align}
    M_{\mu\nu} &\geq \frac{d_\nu}{d_\mu} \left(1-\frac{D-d_\mu}{D^{1-2\tau}}\right),\label{ineq: lb Mmunu AEK 1}
\end{align}
and, moreover,
\begin{align}
    M_{\mu\nu} &\geq 1-\frac{D-d_\nu}{D^{1-2\tau}} \,.\label{ineq: lb Mmunu AEK 2}
\end{align}
\end{thm}

\begin{proof}
Let $\tau>0,\alpha>0,D\geq D_0(\tau,\alpha)$. Since $D_0$ does not depend on the sequence $n^{(D)}$,
\begin{align}
    \mathbb{P}^{(D)}\Bigl(\|\phi_n^{(D)}\|_{\infty} > D^{\tau} D^{-1/2}\Bigr) \leq D^{-\alpha}
\end{align}
for all $n=1,\dots,D$. Writing $\PPP$ for $\PPP^{(D)}$ and $\phi_n$ for $\phi_n^{(D)}$, we obtain that
\begin{subequations}
\begin{align}
    \mathbb{P}\left(\forall n: \|\phi_n\|_{\infty} \leq D^{-1/2+\tau}\right) &= 1-\mathbb{P}\left(\exists n: \|\phi_n\|_{\infty} > D^{-1/2+\tau}\right)\\
    &\geq 1-\sum_{n=1}^D \mathbb{P}\left(\|\phi_n\|_{\infty} > D^{-1/2+\tau}\right)\\
    &\geq 1-D^{-\alpha+1}. 
\end{align}
\end{subequations}
Now assume that $\|\phi_n\|_{\infty} \leq D^{-1/2+\tau}$ for all $n=1,\dots,D$. Then
\begin{subequations}
\begin{align}
    \langle\phi_n|P_\mu|\phi_n\rangle &= 1-\sum_{\mu'\neq \mu} \langle\phi_n|P_{\mu'}|\phi_n\rangle\\
    &= 1- \sum_{l\in [D]\backslash I_\mu} |\phi_n(l)|^2\\
    &\geq 1-\frac{D-d_\mu}{D^{1-2\tau}}\label{ineq: AEK Pmu}
\end{align}
\end{subequations}
where $I_\mu = \{d_1+\dots+d_{\mu-1}+1,\dots,d_1+\dots+d_\mu\}$. 

As in the proof of Theorem~\ref{thm: main new}, we have that
\begin{align}
    |M_{\mu B}| \geq \frac{1}{d_\mu}\left(\sum_n \langle\phi_n|P_\mu|\phi_n\rangle\langle\phi_n|B^{+}|\phi_n\rangle - \sum_m \langle\phi_m|P_\mu|\phi_m\rangle \langle\phi_m|B^{-}|\phi_m\rangle\right),
\end{align}
where $B^{+}$ and $B^-$ denote the positive and negative part of $B$, respectively.

We find that 
\begin{align}
    \sum_n \langle\phi_n|P_\mu|\phi_n\rangle\langle\phi_n|B^{+}|\phi_n\rangle &\geq \tr(B^{+}) \left(1-\frac{D-d_\mu}{D^{1-2\tau}}\right),
\end{align}
and together with \eqref{ineq: Tmain 1}, \eqref{ineq: Tmain 3} and \eqref{ineq: Tmain 4} we obtain that
\begin{align}
    |M_{\mu B}| \geq \max\left\{b_{\min}^{+}, \frac{\tr(B^{+})}{d_\mu}\left(1-\frac{D-d_\mu}{D^{1-2\tau}}\right)\right\} - \min\left\{b_{\max}^{-}, \frac{\tr(B^{-})}{d_\mu}\right\}.
\end{align}
In particular, if $B=P_\nu$ for some macro state $\nu$, then $B^+=P_\nu$, $B^-=0$, $\tr(B^+)=d_\nu$, $b_{\min}^{\pm}=b_{\max}^-=0$, $b_{\max}^+=1$, and thus
\begin{align}
    M_{\mu\nu} \geq \frac{d_\nu}{d_\mu}\left(1-\frac{D-d_\mu}{D^{1-2\tau}}\right).
\end{align}
Alternatively, in the case that $B=P_\nu$ we can apply the estimate in \eqref{ineq: AEK Pmu} to $\langle\phi_n|P_\nu|\phi_n\rangle$, which immediately yields
\begin{align}
    M_{\mu\nu} = \frac{1}{d_\mu}\sum_n \langle\phi_n|P_\mu|\phi_n\rangle \langle\phi_n|P_\nu|\phi_n\rangle \geq 1-\frac{D-d_\nu}{D^{1-2\tau}}.
\end{align}
\end{proof}

If the macro state $\mu$ resp. $\nu$ is such that $(D-d_\mu)D^{-1+2\tau} >1$ resp. $(D-d_\nu)D^{-1+2\tau} >1$, then the lower bound in \eqref{ineq: lb Mmunu AEK 1} resp. \eqref{ineq: lb Mmunu AEK 2} becomes negative and therefore useless since we always have the trivial bound $M_{\mu\nu} \geq 0$. However, if $\mu$ or $\nu$ is the equilibrium macro space in the sense that the corresponding macro space is extremely dominant, more precisely, if $d_{\textup{eq}} = D-o(D^{1-2\tau})$, then the lower bounds for the $M_{\mu\nu}$ are nontrivial. If $\nu=\textup{eq}$, then \eqref{ineq: lb Mmunu AEK 2} implies that $M_{\mu\nu} \gtrsim 1 \approx \frac{d_{\textup{eq}}}{D}$ and if $\mu=\textup{eq}$, then \eqref{ineq: lb Mmunu AEK 1} shows that $M_{\mu\nu} \gtrsim \frac{d_\nu}{d_\mu} \approx \frac{d_\nu}{D}$, in agreement with our expectations.

\section{Consequences of the Eigenstate Thermalization Hypothesis \label{sec: ETH}}

Another result, due to Cipolloni, Erdős and Henheik (CEH) \cite{CEH23}, shows that Wigner matrices (i.e., $H=H^*$ such that $h_{ij}$ for $i\leq j$ are centered, independent random variables with bounded moments, the $h_{ij}$ for $i<j$ are identically distributed, and the $h_{ii}$ are identically distributed) satisfy a version of the eigenstate thermalization hypothesis (ETH) that implies that the eigenvectors are delocalized. Although the matrices we are most interested in, the band matrices, are not Wigner matrices, we make explicit in this section which lower bounds on $M_{\mu\nu}$ (essentially versions of $M_{\mu\nu}\approx d_\nu/D$) would follow from the ETH in the version formulated by CEH (Proposition~\ref{prop: lb ETH}) and what they would imply about the relative error of generalized normal typicality (Corollary~\ref{cor: rel err ETH}). After all, as mentioned in the beginning of Section~\ref{sec: Improved LB}, it is believed that for band matrices with sufficiently wide band, all eigenvectors are delocalized.

We begin by formulating the precise condition:

\begin{defn}[ETH according to CEH \cite{CEH23}]
We say that a sequence $(\PPP^{(D)})_{D\in\mathbb{N}}$ of probability distributions over the Hermitian $D\times D$ matrices satisfies the CEH-version of ETH if for every sequence $(B^{(D)})_{D\in\mathbb{N}}$ of $D\times D$ matrices with $\|B^{(D)}\|\leq 1$ and every $\gamma>0$ and $\xi>0$, there is $\tilde{D}_0\in\mathbb{N}$ such that for $D\geq \tilde{D}_0$, it has probability at least $1-D^{-\gamma}$ that
\begin{align}
    \max_{i,j\in [D]} \biggl|\langle\phi_i^{(D)} | B^{(D)}|\phi_j^{(D)}\rangle - \frac{\tr(B^{(D)})}{D} \delta_{ij}\biggr|  \leq \frac{D^\xi}{D} \tr(|\mathring{B}^{(D)}|^2)^{1/2},\label{ineq: ETH}
\end{align}
where $\phi_1^{(D)},\ldots,\phi_D^{(D)}$ is any orthonormal eigenbasis of $H^{(D)}\sim\PPP^{(D)}$ and $\mathring{B}^{(D)} = B^{(D)} - \tr(B^{(D)})/D$ denotes the traceless part of $B$.
\end{defn}

In that case, we obtain in particular that for any sequence $(B^{(D)})$ of Hermitian $D\times D$ matrices, any $\xi>0$ and $\gamma>0$, sufficiently large $D$, $B=B^{(D)}$ and every orthonormal eigenbasis $\phi_1,\ldots,\phi_D$ of $H=H^{(D)}$, 
\begin{align}
    \langle\phi_n|B|\phi_n\rangle \geq \frac{\tr(B)}{D}-\frac{D^\xi}{D}\tr(|\mathring{B}|^2)^{1/2} \label{ineq: B ETH lb}
\end{align}
for all $n=1,\ldots,D$ simultaneously with probability at least $1-D^{-\gamma}$. Now recall that, if $H$ is non-degenerate, then
\be
   M_{\mu B} = \frac{1}{d_\mu}\sum_n \langle\phi_n|P_\mu|\phi_n\rangle\langle\phi_n|B|\phi_n\rangle\,.
\ee

\begin{prop}[Lower bound for $|M_{\mu B}|$]\label{prop: lb ETH}
Let $\xi>0$, let $H$ be a Hermitian $D\times D$  matrix with orthonormal eigenbasis $\{\phi_1,\ldots,\phi_D\}$, and let $B$ be a Hermitian $D\times D$ matrix such that \eqref{ineq: B ETH lb} is satisfied for all $n=1,\ldots,D$.
Then, for every macro state $\mu$,
\begin{align}
    |M_{\mu B}| \geq \frac{|\tr(B)|}{D}-\frac{D^\xi}{D} \tr(|\mathring{B}|^2)^{1/2}.\label{ineq: MmuB CEK1}
\end{align}
In particular,
\begin{align}
    M_{\mu\nu} \geq \frac{d_\nu}{D}\left(1-\frac{D^\xi }{\sqrt{d_\nu}}\right).
    \label{ineq: MmuB CEK2}
\end{align}
for every macro state $\nu$. Moreover,
\begin{align}
 M_{\mu\nu} \geq \frac{d_\nu}{D}\left(1-\frac{D^\xi }{\sqrt{d_\mu}}\right).
\end{align}
\end{prop}

\begin{proof}
From \eqref{ineq: B ETH lb} we obtain that
\begin{align}
    M_{\mu B} = \frac{1}{d_\mu}\sum_n \langle\phi_n|P_\mu|\phi_n\rangle\langle\phi_n|B|\phi_n\rangle
    \geq\frac{\tr(B)}{D}-\frac{D^\xi}{D} \tr(|\mathring{B}|^2)^{1/2}.
\end{align}
Similarly one finds that
\begin{align}
    M_{\mu B} \leq \frac{\tr(B)}{D}+\frac{D^\xi}{D}\tr(|\mathring{B}|^2)^{1/2}
\end{align}
and therefore
\begin{align}
    \left|M_{\mu B}-\frac{\tr(B)}{D} \right| \leq \frac{D^\xi}{D}\tr(|\mathring{B}|^2)^{1/2}.
\end{align}
With the help of the reverse triangle inequality we obtain \eqref{ineq: MmuB CEK1}. If $B=P_\nu$ for some macro state $\nu$ we have $\tr(B)=d_\nu$, $|\mathring{B}|^2 = P_\nu-2 P_\nu d_\nu/D+d_\nu^2/D^2$, therefore $\tr(|\mathring{B}|^2) = d_\nu-2d_\nu^2/D+d_\nu^2/D^2 \leq d_\nu$ and \eqref{ineq: MmuB CEK2} thus follows immediately from \eqref{ineq: MmuB CEK1}.

Concerning the last bound, observe that \eqref{ineq: B ETH lb} for $B=P_\mu$ implies
\begin{align}
 M_{\mu\nu}= \frac{1}{d_\mu} \sum_n \langle\phi_n|P_\mu|\phi_n\rangle \langle\phi_n|P_\nu|\phi_n\rangle \geq \frac{d_\nu}{d_\mu}\left(\frac{d_\mu}{D}-\frac{D^\xi\sqrt{d_\mu}}{D}\right).
\end{align}    
\end{proof}

For large dimensions $D$ we find for any macro states $\mu$ and $\nu$ that $M_{\mu\nu} \gtrsim \frac{d_\nu}{D}$ (provided that $d_\nu$ is large enough such that $\frac{d_\nu}{D}\gg D^{\xi-1}\sqrt{d_\nu}$), in agreement with our expectations.
For the relative errors of the $\|P_\nu\psi_t\|^2$, 
we obtain the following.

\begin{cor}\label{cor: rel err ETH}
    Suppose that $(\PPP^{(D)})$ satisfies the CEH-version of ETH, and that it is a continuous distribution for every $D$. Let $\varepsilon,\delta,\xi,\gamma>0$, and let $\mu$ and $\nu$ be arbitrary macro states such that $\sqrt{d_\nu}\geq 2 D^\xi$ or $\sqrt{d_\mu}\geq 2 D^\xi$, i.e.,
    \begin{align}
        s_\nu \geq 2\xi s_{\mathrm{mc}} + \frac{2k_B}{N}\ln 2 \quad \mathrm{or}\quad s_\mu \geq 2\xi s_{\mathrm{mc}}+\frac{2k_B}{N}\ln 2.
    \end{align} 
    Then, for sufficiently large $D$ and with probability at least $1-D^{-\gamma}$, $(1-\varepsilon)$-most $\psi_0\in\mathbb{S}(\Hilbert_\mu)$ are such that for $(1-\delta)$-most $t\in[0,\infty)$,
    \begin{align}
        \frac{\bigl|\|P_\nu\psi_t\|^2-M_{\mu\nu} \bigr|}{M_{\mu\nu}} &\leq \frac{8}{\sqrt{\varepsilon\delta}}\exp\left(-\frac{N}{k_B}\left(\max\{s_\mu,s_\nu\}-s_{\mathrm{mc}}+\frac{1}{2}\min\{s_\nu,s_\mu\}\right)\right).
    \end{align}
\end{cor}

\begin{proof}
    It follows immediately from \eqref{ineq: MmuB CEK2} that
    \begin{align}
        M_{\mu\nu} \geq \frac{d_\nu}{2D}.
    \end{align}
    Together with the upper bound for the absolute error from Theorem~\ref{thm: GNT}, we obtain that $(1-\varepsilon)$-most $\psi_0\in\mathbb{S}(\Hilbert_\mu)$ are such that for $(1-\delta)$-most $t\in[0,\infty)$,
    \begin{align}
        \frac{\bigl|\|P_\nu\psi_t\|^2-M_{\mu\nu} \bigr|}{M_{\mu\nu}} \leq \frac{8}{\sqrt{\varepsilon\delta}} \frac{D}{d_\nu\sqrt{d_\mu}} \left(\min\left\{1,\frac{d_\nu}{d_\mu}\right\}\right)^{1/2}
    \end{align}
    The claim now follows immediately from the definition of $s_\mu$, $s_\nu$ and $s_{\mathrm{mc}}$.
\end{proof}

The relative errors in Corollary~\ref{cor: rel err ETH} are small for large $N$ if $s_{\mathrm{mc}}<\max\{s_\mu,s_\nu\}+\min\{s_\mu,s_\nu\}/2$, i.e., we recover the condition we found in the case of normal typicality, see also Remark~\ref{rmk: rel error small}.

\appendix

\section{Appendix - No Resonances \label{app: no res}}

In this appendix we provide a proof of the fact that Hermitian matrices whose joint distribution of their entries is absolutely continuous with respect to the Lebesgue measure have, with probability 1, no degeneracies and no resonances (i.e., also the eigenvalue gaps are non-degenerate). This fact is widely known, but we could not find a proof in the literature.

\begin{prop}[No resonances]\label{prop: no res}
    Let $H$ be a random Hermitian $D\times D$ matrix with eigenvalues $\lambda_1,\dots,\lambda_D$ such that the joint distribution of its entries is absolutely continuous with respect to the Lebesgue measure. Then, 
    \begin{align}
        \mathbb{P}\left(\lambda_i-\lambda_j = \lambda_k-\lambda_l \mbox{ for some } (i\neq j \mbox{ or } k\neq l) \mbox{ and } (i\neq k \mbox{ or } j\neq l)\right)=0.
    \end{align}
\end{prop}

We begin with some preparations for the proof. Let $S(D)$ denote the set of Hermitian $D\times D$ matrices and $U(D)$ the unitary group of $D\times D$ matrices, here regarded as orthonormal bases of the underlying Hilbert space. We define $\chi: \mathbb{R}^D \times U(D) \to S(D)$ by
\begin{align}
    \chi(\lambda_1,\dots,\lambda_D,\psi_1,\dots,\psi_D) = \sum_{j=1}^D \lambda_j |\psi_j\rangle \langle\psi_j|,\label{eq: chi}
\end{align}
where $\psi_j$ denotes the $j$th column of the unitary matrix $[\psi_1,\dots,\psi_D]$. Since the matrix $\chi(\lambda_1,\dots,\lambda_D,\psi_1,\dots,\psi_D)$ remains constant when the phase of the $\psi_i$ is changed, $\chi$ defines a mapping $\varphi: \mathbb{R}^D\times U(D)/U(1)^D \to S(D)$. 
\begin{lemma}\label{lemma: loc diffeo}
The map $\varphi$ defined above is smooth.
\end{lemma}

\begin{proof}
    First note that the function $\chi$ in \eqref{eq: chi} defined on $\mathbb{R}^D\times \mathbb{C}^{D\times D} \to \mathbb{C}^{D\times D}$ (extended in the obvious way) is smooth since all components of $\chi(\lambda_1,\dots,\lambda_D,\psi_1,\dots,\psi_D)$ are polynomials in the $\lambda_j$ and the $(\psi_j)_k$. Since $U(D)$ is a (embedded) submanifold of $\mathbb{C}^{D\times D}$, the restriction of $\chi$ to $\mathbb{R}^D\times U(D)$ is also smooth \cite[Prop.~5.27]{Lee}. It remains to show that also the replacing of $U(D)$ by the quotient $U(D)/U(1)^D$ does not destroy the smoothness. To this end, consider the map
    \begin{align}
        f: \mathbb{C}^D\backslash \{0\} \to \mathbb{C}^{D\times D},\; \psi \mapsto f(\psi) = |\psi\rangle\langle\psi|.
    \end{align}
    This map is obviously smooth and remains smooth when considered as a map on $(\mathbb{C}^D\backslash\{0\})/U(1)$. Thus, we can conclude that $\varphi$ is smooth.
\end{proof}

\begin{proof}[Proof of Proposition \ref{prop: no res}]
    Since $U(1)^D$ is a closed subset of $U(D)$ (because limits of diagonal matrices are diagonal), and since $U(1)^D$ is a Lie subgroup of $U(D)$, it follows that the set $U(D)/U(1)^D$ of left cosets is a smooth manifold \cite[Thm.~21.17]{Lee}.
   It is well known that $\dim U(D)=D^2$, $\dim U(D)/U(1)^D = D^2-D$, $\dim S(D)=D^2$ and thus $\dim \left(\mathbb{R}^D \times (U(D)/U(1)^D)\right) = \dim S(D)$. The spaces $U:=\mathbb{R}^D \times (U(D)/U(1)^D)$ and $V:=S(D)$ are therefore smooth manifolds of the same dimension and $\varphi\big|_{\mathbb{R}^D \times (U(D)/U(1)^D)}: U\to V$ is smooth by Lemma~\ref{lemma: loc diffeo}. 
    Since the set 
    \begin{multline}
        N_{\mathrm{res}} := \Bigl\{(x_1,\dots,x_D)\in\mathbb{R}^D:\\ 
        x_i-x_j = x_k - x_l \mbox{ for some } (i\neq j \mbox{ or } k\neq l) \mbox{ and } (i\neq k \mbox{ or } j\neq l)\Bigr\}
    \end{multline}
    is a null set in $\RRR^D$, it follows that $N_{\mathrm{res}}\times (U(D)/U(1)^D)$ is a null set in $\RRR^D\times (U(D)/U(1)^D)$. Generally, if $f:U\to V$ is a smooth mapping between smooth manifolds $U,V$ of equal dimension and $N\subset U$ is a null set, then $f(N)$ is a null set in $V$ \cite[Thm.~6.9]{Lee}. Thus, 
    \begin{align}
        \varphi(N_{\mathrm{res}}\times (U(D)/U(1)^D)) = \chi(N_{\mathrm{res}} \times U(D))
    \end{align}
    is a null set in $S(D)$. Altogether we have shown that the set of Hermitian matrices that have resonances has measure zero. Finally, the absolute continuity of the joint distribution of the entries of $H$ with respect to the Lebesgue measure immediately proves the claim.
\end{proof}

\section{Appendix - Numerical Examples \label{sec: numerics}}

We briefly describe the numerical   simulations we used to create Figures~\ref{fig:numerical example 1} and \ref{fig:numerical example 2}. These simulations serve for illustrating our results and stating some further conjectures. 
We partition the $D$-dimensional Hilbert space $\Hilbert := \mathbb{C}^D = \bigoplus_{\nu=1}^4\mathbb{C}^{d_\nu}$ into four macro spaces $\Hilbert_\nu$ of dimension $d_\nu$ with $d_\nu\ll d_{\nu+1}$ for $\nu=1,2,3$; more precisely, we choose $d_\nu = 2\times 10^{\nu-1}$. Then $\Hilbert_1$ is spanned by the first $d_1$ canonical basis vectors, $\Hilbert_2$ by the next $d_2$ canonical basis vectors and so on, and $\Hilbert_4$, the largest macro space in the decomposition, corresponds to the ``equilibrium space.'' The Hamiltonian $H=(h_{jk})$ is modelled by a Hermitian random $D\times D$ matrix with a band structure which means that it couples neighboring macro spaces more strongly than distant ones, see also Figure \ref{fig:bandmatrix}. The entries of $H$ satisfy $h_{jj}\sim\mathcal{N}(0,\sigma_{jj}^2)$ and $h_{jk}\sim\mathcal{N}(0,\sigma_{jk}^2/2)+i\mathcal{N}(0,\sigma_{jk}^2/2)$ for $j\neq k$, where
\begin{align}
    \sigma_{jk}^2 := \exp(-s|j-k|)
\end{align}
with some parameter $s>0$ that controls the bandwidth of the random matrix, i.e., the variances decrease exponentially in the distance from the diagonal. Note that this model is not covered by our examples in Section \ref{subsec: ex} since the $\sigma_{jk}$ cannot be bounded by a positive $D$-independent constant or by $D^{\alpha}$ for some $\alpha\in\mathbb{R}$ from below. However, it also suggests that similar results should hold in more general situations than in the ones we were able to study.

\bigskip

\noindent \textbf{Acknowledgments.}
We thank László Erdős and Roman Vershynin for helpful discussions.
C.V.\ gratefully acknowledges financial support by the German Academic Scholarship Foundation. This work was supported by the Deutsche Forschungsgemeinschaft (DFG, German Research Foundation) -- TRR 352 -- Project-ID 470903074.

\bibliographystyle{plainurl}
\bibliography{literature_GNT.bib}

\begin{thebibliography}{10}

\bibitem{AEK17}
O.~H. Ajanki, L.~Erdős, and T.~Krüger.
\newblock {Universality for general Wigner-type matrices}.
\newblock {\em Probability Theory and Related Fields}, 169:667--727, 2017.
\newblock URL: \url{http://arxiv.org/abs/1506.05098}.

\bibitem{AEK19}
O.~H. Ajanki, L.~Erdős, and T.~Krüger.
\newblock {Quadratic Vector Equations On Complex Upper Half-Plane}.
\newblock {\em Memoirs of the American Mathematical Society}, 261, 2019.
\newblock URL: \url{http://arxiv.org/abs/1506.05095}.

\bibitem{BRGSR18}
B.~Balz, J.~Richter, J.~Gemmer, R.~Steinigeweg, and P.~Reimann.
\newblock {Dynamical typicality for initial states with a preset measurement
  statistics of several commuting observables}.
\newblock In F.~Binger, L.~A. Correa, C.~Gogolin, J.~Anders, and G.~Adesso,
  editors, {\em {Thermodynamics in the Quantum Regime}}, chapter~17, pages
  413--433. Springer, Cham, 2019.
\newblock URL: \url{http://arxiv.org/abs/1904.03105}.

\bibitem{BG09}
C.~Bartsch and J.~Gemmer.
\newblock Dynamical typicality of quantum expectation values.
\newblock {\em Physical Review Letters}, 102:110403, 2009.
\newblock URL: \url{http://arxiv.org/abs/0902.0927}.

\bibitem{BYY20}
P.~Bourgade, H.-T. Yau, and J.~Yin.
\newblock {Random Band Matrices in the Delocalized Phase I: Quantum Unique
  Ergodicity and Universality}.
\newblock {\em Communications on Pure and Applied Mathematics}, 73(7), 2020.
\newblock URL: \url{http://arxiv.org/abs/1807.01559}.

\bibitem{CEH23}
G.~Cipolloni, L.~Erdős, and J.~Henheik.
\newblock {Eigenstate thermalisation at the edge for Wigner matrices}, 2023.
\newblock Preprint.
\newblock URL: \url{https://arxiv.org/abs/2309.05488}.

\bibitem{CEHK23}
G.~Cipolloni, L.~Erdős, J.~Henheik, and O.~Kolupaiev.
\newblock {Gaussian fluctuations in the Equipartition Principle for Wigner
  matrices}, 2023.
\newblock Preprint.
\newblock URL: \url{https://arxiv.org/abs/2301.05181}.

\bibitem{Erdoes21}
G.~Cipolloni, L.~Erdős, and D.~Schröder.
\newblock {Eigenstate Thermalization Hypothesis for Wigner Matrices}.
\newblock {\em Communications in Mathematical Physics}, 388:1005--1048, 2021.
\newblock URL: \url{http://arxiv.org/abs/2012.13215}.

\bibitem{Deutsch91}
J.~M. Deutsch.
\newblock {Quantum statistical mechanics in a closed system}.
\newblock {\em Physical Review A}, 43:2046--2049, 1991.

\bibitem{EK11b}
L.~Erdős and A.~Knowles.
\newblock {Quantum diffusion and delocalization for band matrices with general
  distribution}.
\newblock {\em Annales de l'Institut Henri Poincaré}, 12:1227--1319, 2011.
\newblock URL: \url{http://arxiv.org/abs/1005.1838}.

\bibitem{EK11}
L.~Erdős and A.~Knowles.
\newblock {Quantum diffusion and eigenfunction delocalization in a random
  matrix band model}.
\newblock {\em Communications in Mathematical Physics}, 303:509--554, 2011.
\newblock URL: \url{http://arxiv.org/abs/1002.1695}.

\bibitem{EKS19}
L.~Erdős, T.~Krüger, and D.~Schröder.
\newblock {Random Matrices with Slow Correlation Decay}.
\newblock {\em Forum of Mathematics, Sigma}, 7:E8, 2019.
\newblock URL: \url{https://arxiv.org/abs/1705.10661}.

\bibitem{ESY09}
L.~Erdős, B.~Schlein, and H.-T. Yau.
\newblock {Local semicircle law and complete delocalization for Wigner random
  matrices}.
\newblock {\em Communications in Mathematical Physics}, 287:641--655, 2009.
\newblock URL: \url{http://arxiv.org/abs/0803.0542}.

\bibitem{ESY09b}
L.~Erdős, B.~Schlein, and H.-T. Yau.
\newblock {Semicircle law on short scales and delocalization of eigenvectors
  for Wigner random matrices}.
\newblock {\em Annals of Probability}, 37:815--852, 2009.
\newblock URL: \url{http://arxiv.org/abs/0711.1730}.

\bibitem{EYY12}
L.~Erdős, H.-T. Yau, and J.~Yin.
\newblock {Bulk universality for generalized Wigner matrices}.
\newblock {\em Probability Theory and Related Fields}, 154:341--407, 2012.
\newblock URL: \url{http://arxiv.org/abs/1001.3453}.

\bibitem{GHT13}
S.~Goldstein, T.~Hara, and H.~Tasaki.
\newblock {Time Scales in the Approach to Equilibrium of Macroscopic Quantum
  Systems}.
\newblock {\em Physical Review Letters}, 111:140401, 2013.
\newblock URL: \url{http://arxiv.org/abs/1307.0572}.

\bibitem{GHT14}
S.~Goldstein, T.~Hara, and H.~Tasaki.
\newblock The approach to equilibrium in a macroscopic quantum system for a
  typical nonequilibrium subspace, 2014.
\newblock Preprint, \url{http://arxiv.org/abs/1402.3380}.

\bibitem{GHT15}
S.~Goldstein, T.~Hara, and H.~Tasaki.
\newblock {Extremely quick thermalization in a macroscopic quantum system for a
  typical nonequilibrium subspace}.
\newblock {\em New Journal of Physics}, 17:045002, 2015.
\newblock URL: \url{http://arxiv.org/abs/1402.0324}.

\bibitem{GLMTZ09}
S.~Goldstein, J.~L. Lebowitz, C.~Mastrodonato, R.~Tumulka, and N.~Zangh\`\i.
\newblock {Normal Typicality and von Neumann's Quantum Ergodic Theorem}.
\newblock {\em Proceedings of the Royal Society A}, 466(2123):3203--3224, 2010.
\newblock URL: \url{http://arxiv.org/abs/0907.0108}.

\bibitem{GLMTZ10}
S.~Goldstein, J.~L. Lebowitz, C.~Mastrodonato, R.~Tumulka, and N.~Zangh\`\i.
\newblock {On the Approach to Thermal Equilibrium of Macroscopic Quantum
  Systems}.
\newblock {\em Physical Review E}, 81:011109, 2010.
\newblock URL: \url{http://arxiv.org/abs/0911.1724}.

\bibitem{GLTZ10}
S.~Goldstein, J.~L. Lebowitz, R.~Tumulka, and N.~Zangh\`\i.
\newblock {Long-Time Behavior of Macroscopic Quantum Systems}.
\newblock {\em European Physical Journal H}, 35:173--200, 2010.
\newblock URL: \url{http://arxiv.org/abs/1003.2129}.

\bibitem{Latala05}
R.~Latala.
\newblock {Some Estimates of Norms of Random Matrices}.
\newblock {\em Proceedings of the American Mathematical Society},
  133:1273--1282, 2005.

\bibitem{Leb93}
J.L. Lebowitz.
\newblock Macroscopic laws, microscopic dynamics, time's arrow and
  {B}oltzmann's entropy.
\newblock {\em Physica A}, 194:1--27, 1993.

\bibitem{Lee}
J.~M. Lee.
\newblock {\em Introduction to Smooth Manifolds, 2nd ed.}
\newblock Springer, 2013.

\bibitem{LR20}
K.~Luh and S.~O'Rourke.
\newblock {Eigenvector delocalization for non-Hermitian random matrices and
  applications}.
\newblock {\em Random Structures \& Algorithms}, 57:169--210, 2020.
\newblock URL: \url{http://arxiv.org/abs/1810.00489}.

\bibitem{LT20}
A.~Lytova and K.~Tikhomirov.
\newblock {On delocalization of eigenvectors of random non-Hermitian matrices}.
\newblock {\em Probability Theory and Related Fields}, 177:465--524, 2020.
\newblock URL: \url{http://arxiv.org/abs/1810.01590}.

\bibitem{MGE}
M.~P. M\"uller, D.~Gross, and J.~Eisert.
\newblock Concentration of measure for quantum states with a fixed expectation
  value.
\newblock {\em Communications in Mathematical Physics}, 303:785--824, 2011.
\newblock URL: \url{http://arxiv.org/abs/1003.4982}.

\bibitem{Normalnumber}
{Normal number}.
\newblock In \textit{Wikipedia, the free encyclopedia} (accessed 12/12/2021).
\newblock URL: \url{http://en. wikipedia.org/wiki/Normal_number}.

\bibitem{Reimann07}
P.~Reimann.
\newblock {Typicality for Generalized Microcanonical Ensembles}.
\newblock {\em Physical Review Letters}, 99:160404, 2007.
\newblock URL: \url{http://arxiv.org/abs/0710.4214}.

\bibitem{Reimann08}
P.~Reimann.
\newblock {Foundations of Statistical Mechanics under Experimentally Realistic
  Conditions}.
\newblock {\em Physical Review Letters}, 101:190403, 2008.
\newblock URL: \url{http://arxiv.org/abs/0810.3092}.

\bibitem{Reimann2015}
P.~Reimann.
\newblock Generalization of von {N}eumann's approach to thermalization.
\newblock {\em Physical Review Letters}, 115:010403, 2015.
\newblock URL: \url{http://arxiv.org/abs/1507.00262}.

\bibitem{Reimann2018b}
P.~Reimann.
\newblock Dynamical typicality approach to eigenstate thermalization.
\newblock {\em Physical Review Letters}, 120:230601, 2018.
\newblock URL: \url{http://arxiv.org/abs/1806.03193}.

\bibitem{Reimann2018a}
P.~Reimann.
\newblock Dynamical typicality of isolated many-body quantum systems.
\newblock {\em Physical Review E}, 97:062129, 2018.
\newblock URL: \url{http://arxiv.org/abs/1805.07085}.

\bibitem{RG20}
P.~Reimann and J.~Gemmer.
\newblock Why are macroscopic experiments reproducible? {I}mitating the
  behavior of an ensemble by single pure states.
\newblock {\em Physica A}, 552:121840, 2020.
\newblock URL: \url{http://arxiv.org/abs/2005.14626}.

\bibitem{RV15}
M.~Rudelson and R.~Vershynin.
\newblock {Delocalization of eigenvectors of random matrices with independent
  entries}.
\newblock {\em Duke Mathematical Journal}, 164:2507--2538, 2015.
\newblock URL: \url{http://arxiv.org/abs/1306.2887}.

\bibitem{RV15b}
M.~Rudelson and R.~Vershynin.
\newblock {Small Ball Probabilities for Linear Images of High-Dimensional
  Distributions}.
\newblock {\em International Mathematics Research Notices},
  2015(19):9594--9617, 2015.
\newblock URL: \url{http://arxiv.org/abs/1402.4492}.

\bibitem{RV16}
M.~Rudelson and R.~Vershynin.
\newblock {No-Gaps Delocalization for General Random Matrices}.
\newblock {\em Geometric and Functional Analysis}, 26:1716--1776, 2016.
\newblock URL: \url{http://arxiv.org/abs/1506.04012}.

\bibitem{Short11}
A.~J. Short.
\newblock Equilibration of quantum systems and subsystems.
\newblock {\em New Journal of Physics}, 13:053009, 2011.
\newblock URL: \url{http://arxiv.org/abs/1012.4622}.

\bibitem{SF12}
A.~J. Short and T.~C. Farrelly.
\newblock Quantum equilibration in finite time.
\newblock {\em New Journal of Physics}, 14:013063, 2012.

\bibitem{Srednicki94}
M.~Srednicki.
\newblock {Chaos and quantum thermalization}.
\newblock {\em Physical Review E}, 50:888--901, 1994.
\newblock URL: \url{http://arxiv.org/abs/cond-mat/9403051}.

\bibitem{SW71}
E.~M. Stein and G.~Weiss.
\newblock {\em {Introduction to Fourier analysis on Euclidean spaces
  (PMS-32)}}.
\newblock Princeton University Press, 1971.

\bibitem{SWGW22}
P.~Strasberg, A.~Winter, J.~Gemmer, and J.~Wang.
\newblock Classicality, {M}arkovianity, and local detailed balance from pure
  state dynamics.
\newblock Preprint, 2022.
\newblock URL: \url{http://arxiv.org/abs/2209.07977}.

\bibitem{TTV22-physik}
S.~Teufel, R.~Tumulka, and C.~Vogel.
\newblock {Time evolution of typical pure states from a macroscopic Hilbert
  subspace}.
\newblock {\em Journal of Statistical Physics}, 190:69, 2023.
\newblock URL: \url{http://arxiv.org/abs/2210.10018}.

\bibitem{gen-can}
S.~Teufel, R.~Tumulka, and C.~Vogel.
\newblock {Canonical Typicality for Other Ensembles than Micro-Canonical}.
\newblock {\em Annales Henri Poincar\'e}, 2024.
\newblock URL: \url{http://arxiv.org/abs/2307.15624}.

\bibitem{vonNeumann29}
J.~von Neumann.
\newblock {Beweis des Ergodensatzes und des $H$-Theorems in der neuen
  Mechanik}.
\newblock {\em Zeitschrift f\"ur Physik}, 57:30--70, 1929.
\newblock English translation: \textit{European Physical Journal H}, 35:
  201--237, 2010.
\newblock URL: \url{http://arxiv.org/abs/1003.2133}.

\bibitem{Wigner55}
E.~P. Wigner.
\newblock {Characteristic Vectors of Bordered Matrices With Infinite
  Dimensions}.
\newblock {\em Annals of Mathematics}, 62(3):548--564, 1955.

\bibitem{YYY21}
F.~Yang, H.-T. Yau, and J.~Yin.
\newblock {Delocalization and quantum diffusion of random band matrices in high
  dimensions I: Self-energy renormalization}, 2021.
\newblock Preprint.
\newblock URL: \url{http://arxiv.org/abs/2104.12048}.

\end{thebibliography}

\end{document}